\documentclass[11pt]{article}
\usepackage[margin=1in]{geometry}
\usepackage{graphicx}
\usepackage{amsmath, amsthm, amssymb, amsfonts}
\usepackage{mathtools}
\mathtoolsset{mathic=true}
\mathtoolsset{centercolon}

\usepackage{bbm}

\usepackage{xspace}
\usepackage{hyperref}
\usepackage{todonotes}
\usepackage{algorithm}
\usepackage[noend]{algpseudocode}
\usepackage{caption}
\usepackage{subcaption}
\usepackage{tikz}
\usetikzlibrary{math}
\usetikzlibrary{arrows.meta,decorations.pathreplacing}

\usepackage{bm}
\usepackage{enumitem}
\usepackage{booktabs}
\usepackage{glossaries-extra}
\glsdisablehyper
\glssetcategoryattribute{initialism}{accessinsertdots}{true}

\theoremstyle{plain}
\newtheorem{theorem}{Theorem}[section]
\newtheorem{lemma}[theorem]{Lemma}
\newtheorem{proposition}[theorem]{Proposition}

\newtheorem{corollary}{Corollary}
\theoremstyle{definition}
\newtheorem{definition}{Definition}[section]

\theoremstyle{remark}
\newtheorem*{remark}{Remark}

\newcommand{\ignore}[1]{}
\newcommand{\rmk}[1]{}
\newcommand{\unset}{\mathcal{P}}

\newcommand{\kofn}{$k$-of-$n$\xspace}
\DeclareMathOperator{\E}{\mathbb{E}}

\newcommand{\opt}{\mathtt{OPT}}
\newcommand{\alg}{\mathtt{ALG}}
\DeclarePairedDelimiter\abs{\lvert}{\rvert}%
\DeclarePairedDelimiter\norm{\lVert}{\rVert}%
\newcommand{\Z}{\mathbb{Z}}
\newcommand{\indicator}[1]{\mathbbm{1}\mathopen{}\left(#1\right)\mathclose{}}
\DeclareMathOperator{\pobin}{\mathsf{PBD}}
\DeclareMathOperator{\bern}{\mathsf{Ber}}

\DeclarePairedDelimiter{\ceil}{\lceil}{\rceil}
\DeclarePairedDelimiter{\floor}{\lfloor}{\rfloor}
\newcommand{\stg}{\nu}

\DeclareMathOperator{\adv}{ADV}
\DeclareMathOperator{\advb}{\overline{ADV}}
\newcommand{\bigOh}[1]{O\mathopen{}\left(#1\right)\mathclose{}}

\renewcommand{\vec}[1]{\boldsymbol{\mathbf{#1}}}
\newcommand{\reals}{\mathbb{R}}

\newcommand{\probb}{\ensuremath{\overline{\mathtt{ADV}}}\xspace}

\newcommand{\pbrv}{\mathtt{PBD}}
\DeclareMathOperator{\tvdist}{d_{\mathrm{TV}}}
\DeclareMathOperator{\poly}{poly}

\newcommand{\tn}{\tilde N}
\newcommand{\asol}{\pi}
\newcommand{\osol}{\sigma}
\newcommand{\sol}{\sigma}
\newcommand{\vomega}{\vec \omega}
\newcommand{\normaldist}{\mathcal N}

\makeatletter
\let\oldabs\abs
\def\abs{\@ifstar{\oldabs}{\oldabs*}}
\let\oldnorm\norm
\def\norm{\@ifstar{\oldnorm}{\oldnorm*}}
\makeatother

\newcommand{\RN}[1]{%
  \textup{\uppercase\expandafter{\romannumeral#1}}%
}

\newcommand{\sst}{\ensuremath{\texttt{SST}}}

\newacronym[category=initialism]{rv}{r.v.}{random variable}
\newacronym[category=initialism]{pbd}{PBD}{Poisson binomial distribution}
\newacronym[category=initialism]{tv}{TV}{total variation}
\newacronym{drst}{\ensuremath{\mathtt{DRST}}\xspace}{distributionally robust sequential testing problem}
\newacronym[category=initialism]{qptas}{QPTAS}{quasi-polynomial-time approximation scheme}
\newacronym[category=initialism]{adv}{\ensuremath{\mathtt{ADV}}\xspace}
{adversary's problem}

\title{Distributionally Robust \(k\)-of-\(n\) Sequential Testing}
\author{Rayen Tan\thanks{Department of Industrial and Operations Engineering, University of Michigan, Ann Arbor, USA. Research supported in part by NSF grant  CCF-2418495.} \and Viswanath Nagarajan$^*$}

\begin{document}
\maketitle

\begin{abstract}
    The $k$-of-$n$ testing problem involves performing $n$ independent tests sequentially, in order to determine whether/not at least $k$ tests pass.
    The objective is to minimize the expected  cost of testing. This is a fundamental and well-studied stochastic optimization problem. However, a key limitation of this model is that the success/failure probability of each test is assumed to be known precisely. In this paper, we relax this assumption and study a distributionally-robust model for \kofn testing. In our setting, each test is associated with an interval that contains its  (unknown) failure probability. The goal is to find a solution that minimizes the worst-case expected cost, where each test's  probability is chosen from its interval.  We focus on non-adaptive solutions, that are specified by a fixed permutation of the tests. When all test costs are unit, we obtain a $2$-approximation algorithm for distributionally-robust \kofn testing.     For general costs, we obtain  an $O(\frac{1}{\sqrt \epsilon})$-approximation algorithm 
    on $\epsilon$-bounded instances where   each uncertainty interval  is contained in $[\epsilon, 1-\epsilon]$.  We also consider  the inner maximization problem for distributionally-robust \kofn: this involves finding the worst-case probabilities from the uncertainty intervals for a given solution. For this problem, in addition to the above approximation ratios, we obtain  a quasi-polynomial time approximation scheme under the assumption that all costs are polynomially bounded. 
\end{abstract}

\section{Introduction} \label{sec:intro}
Sequential testing involves determining the status of a complex system by performing a sequence of tests. The test outcomes (pass or fail) are random  and typically assumed to be independent of each other. It is also assumed that we know the failure probabilities of each test.
The goal is to minimize the expected cost of tests performed until the system status is determined. Sequential testing problems arise in many applications, such as manufacturing, healthcare and databases; see e.g., the surveys \cite{moret_decision_1982,unluyurt_sequential_2004,unluyurt_sequential_2025}.
In this paper, we study a fundamental sequential testing problem, \kofn testing, where one wants to determine if at least $k$ out of $n$ tests pass. 

There are two types of solutions in sequential testing. A non-adaptive solution involves performing tests in an {\em a priori} fixed order until the system status is determined. On the other hand, an adaptive solution is a sequential decision process, that at each step chooses the next test based on the outcomes of all previous tests.  Although adaptive solutions can achieve a better objective value, non-adaptive solutions are simpler and faster to implement in practice. A non-adaptive solution is   computed  entirely upfront 
(before testing begins) and the incremental work between successive tests is minimal (one just needs to check if the status has been determined).
We focus on non-adaptive solutions in this paper.

A significant drawback of  sequential testing models (and stochastic optimization, more generally) is that one needs access to a precise probability distribution.
This is  unrealistic because the underlying probabilities are usually obtained through statistical methods that provide confidence intervals rather than point-estimates.
Therefore, we are  interested in  sequential testing algorithms that are  ``robust'' to small changes in the probability distribution. We address this issue through the framework of distributionally robust optimization (DRO), see e.g., the surveys \cite{RahimianM22,kuhn_distributionally_2025}. Here, we are given an ambiguity/uncertainty set $\unset$ of probability distributions (instead of a single distribution) and the goal is to find a solution that minimizes the worst-case expected cost, over all distributions in $\unset$.  In our paper, we consider a natural ``hypercube'' uncertainty set, where the failure probability of each test lies in some  uncertainty interval.

\def\sse{\subseteq}

\subsection{Problem Definition} \label{sec:problem-def}
\paragraph{\kofn testing.} An instance of this basic problem consists of $n$ independent tests with known success probabilities $\{p_i\}_{i=1}^n$ and costs $\{c_i\}_{i=1}^n$, and a threshold $k\le n$. The outcome of any test  $i\in [n]:=\{1,2,\dots, n\}$ is denoted by  a Bernoulli \gls{rv}  $X_i\sim\bern(p_i)$, where $X_i = 1$ when test $i$ passes and $X_i = 0$ when test $i$ fails. In order to determine the outcome of test $i$ we need to perform the test, which  incurs cost $c_i$. 
  When at least $k$ out of the $n$ tests pass, we say that the system \emph{passes}; otherwise, the system \emph{fails}.
Formally, we want to evaluate the function $f\colon 2^{n} \mapsto \{0, 1\}$  given by $f(X_1, X_2, \dots, X_n) = \indicator{\sum_{i=1}^n X_i \geq k}$. Note that this equals  1 when the system passes and is 0 otherwise. Our goal is to perform tests sequentially (one by one) until $f$ is evaluated. The objective is to minimize the {\em expected} total cost of testing.  It is important to note that evaluating   function $f$ does not require knowing the values of all $X_i$s. In particular, when some subset $S\sse [n]$ of tests have been performed we can stop testing if either of the following holds:
\begin{itemize}
    \item At least $k$ tests in $S$ passed, i.e., $\sum_{i\in S}X_i\ge k$.
    \item At least $n-k+1$ tests in $S$ failed, i.e., $\sum_{i\in S}X_i\le \abs{S}- n+ k-1$.
\end{itemize}
We refer to these as  \emph{stopping conditions} for the \kofn problem.

There are two types of solutions for \kofn testing. A {\em non-adaptive} solution is given by a permutation $\sol=\langle \sol_1,\dots, \sol_n\rangle$ of the $n$ tests, where tests are performed in that order until the stopping condition.  An adaptive solution, by contrast, selects the next test to perform based on the outcomes of previous tests. While adaptive solutions are better in terms of the objective value, non-adaptive solutions are much simpler to implement in practice. In this paper, we focus on non-adaptive solutions.

The expected cost of a non-adaptive solution $\sol$ is given by
\begin{equation} \label{eq:cost-fn}
  \sum_{i=1}^n c_{\sol_i} \times \Pr\left[ i - 1-(n - k) \leq \sum_{j=1}^{i-1} X_{\sol_j} \le k-1 \right],
\end{equation}
where the probability term can be interpreted as the probability that the $i^{th}$ test  needs to be performed.

\paragraph{The distributionally robust (DR) setting.} Here, we do not know the exact probabilities $p_i$s associated with the tests. Instead, for each test $i\in [n]$ we are only given an {\em uncertainty interval} $[\ell_i,r_i]$  that contains the true probability $p_i$.  We define the \textit{uncertainty set} $\unset$ to be the set of all probability vectors $\vec{p}\in [0,1]^n$ that respect these intervals, i.e.,
\[ \unset = \left\{\vec p \in [0, 1]^n: \ell_i \leq p_i \leq r_i \text{ for } i = 1, 2, \dots ,n\right\}.\]
Since $\vec{p} \in \unset$ is not known in the distributionally robust (DR) setting, we want to guard against the worst-case choice of $\vec{p}$. Formally, the  objective value of a non-adaptive solution $\sol$ in the DR setting is 
$$\max_{\vec{p}\in \unset} C(\sol, \vec{p}),$$ 
where $C(\sol, \vec{p})$ is given by~\eqref{eq:cost-fn}, the expected cost of $\sol$ under test-probabilities $\vec p$.

We consider two natural problems in the DR setting. The first problem (Definition~\ref{def:adv-defn}) involves evaluating the DR objective value of a given non-adaptive solution and the second problem (Definition~\ref{def:drst-defn}) involves finding a non-adaptive solution with minimum DR objective.  One way to think about the DR setting  is that whenever we pick some solution $\sol$, an ``adversary'' would set the distribution $\vec{p}\in \unset$ so as to maximize the expected cost of testing.

\begin{definition}[\kofn Adversary  problem, $\mathtt{ADV}$] \label{def:adv-defn}
    Given a non-adaptive solution $\sol$, find the probability-vector  $\vec{p} \in \unset$ that maximizes $C(\sol,\vec p)$. That is,
    \begin{equation}
        \adv(\sol) = \max_{\vec{p}\in \unset} C(\sol, \vec{p}).
    \end{equation}
\end{definition}
The function $C(\sol, \vec p)$ in \eqref{eq:cost-fn} is neither convex nor concave in $\vec p$. Consequently, solving the adversarial problem is itself quite challenging.  We observe that $C(\sol, \vec p)$   is multilinear in $\vec p$, which implies that the ``optimal''  probability vector $\vec p \in \unset$ occurs at  an extreme point (i.e., each $p_i$  equals $\ell_i$ or $r_i$).

\begin{definition}[Distributionally Robust \kofn, $\mathtt{DRST}$]\label{def:drst-defn}
Find  the permutation $\sol=\langle \sol_1,\dots, \sol_n\rangle$ that minimizes the worst-case cost of testing, i.e., 
    \begin{equation}
    \min_{\sol} \,\, \max_{\vec p\in \unset} \,\, C(\sol, \vec{p}) \quad=\quad \min_\sol\,\, \adv(\sol).
\end{equation}
\end{definition}

This problem is a bi-level optimization problem, where the first level involves the permutation $\sol$ and the second level involves the probability-vector $\vec p$.

\paragraph{Bounded instances.} Some of our  results rely on the following boundedness assumption on the uncertainty set.   
\begin{definition}\label{def:eps-bounded} For any $\epsilon\in (0,\frac12)$, an $\epsilon$-bounded instance of  $\mathtt{ADV}$  and $\mathtt{DRST}$ is one where every uncertainty interval is contained in $[\epsilon,1-\epsilon]$, i.e., $\epsilon \le \ell_i\le r_i\le 1-\epsilon$ for all $i\in [n]$.     
\end{definition}
This is a reasonable assumption  because tests with probabilities  very close to $0$ or $1$ can be viewed as having deterministic outcomes. Moreover, our algorithms  do not need an explicit bound on $\epsilon$, although the approximation ratio degrades as  $\epsilon\rightarrow 0$.

\subsection{Result and techniques}\label{sec:res-and-techniques}

Our first result is for the special case when all costs are uniform. 
\begin{theorem} \label{thm:unit-cost}
There is a $2$-approximation algorithm for the \kofn adversary  problem and the   distributionally robust \kofn  problem under unit costs.
\end{theorem}
This algorithm relies on a useful property of the unit-cost setting: one of the two stopping conditions will occur only  after the cost incurred is at least $\frac{n}2$, which is a constant fraction of the total cost $n$.   So, at the loss of factor two in the approximation, the adversary can ignore one stopping condition, and optimize for the other. Optimizing for just one of the two stopping conditions turns out to be simple: the adversary will either set $\vec p = \vec \ell$ or  $\vec p = \vec r$.
The algorithm for the distributionally robust problem  is  also based on  this characterization: it   selects tests in the order of either $\ell_i$ or $r_i$.

Our second result handles instances with general costs, under  the assumption that all  probabilities are bounded away from $0$ and $1$. 

\begin{theorem} \label{thm:gen-cost}
There is an  $\bigOh{\frac1{\sqrt{\epsilon}}}$-approximation algorithm for   the \kofn adversary  problem and the   distributionally robust \kofn  problem on $\epsilon$-bounded instances with $\epsilon=\Omega(\frac1n)$.
\end{theorem}

This result relies on a careful analysis of the ``non completion'' probability at each stage $\nu\in [n]$, i.e., the  probability that the testing continues past the $\nu^{th}$ test.
For each stage $\nu \in [n]$, the number of passes observed follows a \gls{pbd}.
We then use various concentration and anti-concentration properties of \glspl{pbd} to simplify the \gls{adv}, at the loss of a $\frac1{\sqrt{\epsilon}}$ factor. In particular, instead of  having to  choose each probability $\{p_i\}_{i=1}^n$, we show that it suffices to carefully choose the ``prefix mean''   $\sum_{i=1}^\stg{p_i}$ for each stage $\stg$.

Finally, we provide an improved approximation algorithm for the adversary problem at the expense of more computational time.
Recall that quasi-polynomial running time  is of the form $n^{\log^c n}$ for constant $c$, where $n$ is the instance size. This is smaller  than exponential time $2^{\poly(n)}$ but larger  than polynomial time.  
\begin{theorem} \label{thm:qptas-adv}
    There is a \gls{qptas} for the \kofn adversary problem assuming that the aspect ratio $\frac{c_{\max}}{c_{\min}}$ is polynomial in $n$.
\end{theorem}
This algorithm uses a structural result on \glspl{pbd} from \cite{daskalakis_sparse_2015}, which states that any two 
\glspl{pbd} having equal first $d$ moments are very close   
 in  total-variation distance  ($\approx 2^{-d/2}$). We need to set $d\approx \log n$ in order to achieve a good enough approximation. Finally, we need to apply dynamic programming on top of this structural result, which leads to an exponential dependence on $d$.

\subsection{Related Work}\label{sec:rel-works}
The classical sequential testing problem was introduced in the quality engineering literature in \cite{butterworth_reliability_1972,mitten_1960}, where ``series'' systems  are studied (this is a special case of \kofn where all $n$ tests need to pass). It is well known that the greedy algorithm that selects tests in increasing order of $\frac{c_i}{1-p_i}$ is optimal for series systems. Moreover, there is no difference between adaptive and non-adaptive solutions. 

For the more general \kofn problem, there is an exact adaptive algorithm \cite{ben-dov_optimal_1981} that involves carefully interleaving two lists: increasing order of $\frac{c_i}{1-p_i}$ and $\frac{c_i}{p_i}$ respectively. The non-adaptive setting turns out to be more challenging here. A round-robin combination of these two lists was  used in \cite{gkenosisStochasticScoreClassification2018a} to obtain a $2$-approximation algorithm for non-adaptive \kofn, even in comparison to the optimal adaptive cost. Better results are known in the {\em unit-cost} setting: \cite{grammel2022algorithms} obtained a $\frac32$-approximation algorithm and  \cite{nielsen2025nonadaptiveevaluationkofnfunctions} obtained a polynomial-time approximation scheme (PTAS).
Tight  ``adaptivity gaps'' (worst case ratio between optimal non-adaptive and adaptive solutions) are also known for the \kofn problem: $\frac32$ for unit-costs \cite{nielsen2025nonadaptiveevaluationkofnfunctions} and $2$ for general costs \cite{gkenosisStochasticScoreClassification2018a}.   

Approximation algorithms have been obtained for several other sequential testing problems, such as halfspace evaluation~\cite{deshpande2016approximation,ghugeNonadaptiveStochasticScore2025}, score classification~\cite{gkenosisStochasticScoreClassification2018a, ghugeNonadaptiveStochasticScore2025,PlankS24}, symmetric boolean functions~\cite{gkenosis2022stochastic,nielsen2025nonadaptiveevaluationkofnfunctions}, voting functions~\cite{hellerstein2024quickly,keles2026exact} and basis testing in partition matroids~\cite{hellerstein2026approximating}.   All the above results rely on knowing  exact success/failure probabilities. In this paper, we relax this assumption and study a distributionally-robust model    for \kofn testing.

There is also a line of work that explores sequential testing in settings where costs are not additive (as above) but subadditive, which models economies of scale.  
\cite{daldal_sequential_2017} formulated the batched testing problem, where a fixed-cost is charged for conducting a batch of tests.
\cite{daldal_approximation_2016} and \cite{segev_polynomial-time_2022} obtained approximation algorithms for series-systems under batch costs, the latter being a PTAS. 
 \cite{tanGeneralFrameworkSequential2025} obtained a general framework to transform any non-adaptive algorithm for sequential testing with additive cost to  batch costs; this provides a $2.62$-approximation algorithm for \kofn testing under batch costs. More recently,  \cite{harrisSequentialTestingSubadditive2025a} obtained approximation algorithms for series-systems under very general subadditive cost structures.

Distributionally-robust optimization (DRO) is itself a classical framework in Operations Research, proposed in the 1950s by \cite{scarf1957min}.  For a recent review, we refer readers to \cite{kuhn_distributionally_2025}; most of the existing works in DRO focus on continuous optimization problems.
There are some DRO results for discrete optimization as well.
For example, \cite{bertsimas2004probabilistic} studied DRO for combinatorial optimization problems with random objective coefficients and moment-based uncertainty sets; these can be viewed as DRO for single-stage decision problems. \cite{agrawal2012price} obtained approximation algorithms for DRO on a class of single-stage stochastic problems with correlated distributions having a prescribed mean.    
More recently, \cite{linhares2019approximation} obtained approximation algorithms for two-stage stochastic covering problems such as facility location and set cover. In contrast to these previous results, we  address {\em  multistage}  stochastic optimization in the DRO framework.

\subsection{Preliminaries}\label{sec:preliminaries}
In this section, we first summarize key properties of the Poisson Binomial Distribution, a generalization of the binomial distribution with non-identical success probabilities. Next, we introduce a dynamic program approach for the DR \kofn problem that is used crucially in all our results.

\subsubsection{Poisson Binomial Distribution (PBD)}\label{sec:pbd}
We let $S = \sum_{i=1}^n X_i$ be the \gls{rv}\@ denoting the number of passes among the $n$ tests.
The \gls{rv}\@ $S$ follow a \gls{pbd}.
We say that a \gls{pbd} is of \emph{order} $n$ if it is a sum of $n$ independent Bernoulli \glspl{rv}.
We use the shorthand $\pbrv(\vec p)$ to represent an \gls{rv}\@ that follows a \gls{pbd} with probabilities $\vec p$. That is, $ \pbrv(\vec p) = \sum_{i=1}^n X_i $ where $X_i \sim \bern(p_i)$.
We sometimes want to know the number  of passes only within the first $\nu$ tests.
Thus, we introduce the notation $\pbrv_\stg(\vec p) = \sum_{i=1}^\nu X_i$ to denote the number of passes among the prefix $[\nu]$ of tests.\footnote{Note that $\pbrv(\vec p)$ only depends on the first $\stg$ probabilities in $\vec p$.} 

\begin{definition}[Poisson binomial distribution] \label{def:pbd}
    Let $\{X_i\}_{i=1}^n$ be a set of independent Bernoulli \glspl{rv} with probability $p_i$ of success.
    The \gls{rv}\@ $\pbrv(\vec p)$ has the probability mass function
    \begin{equation*}
        \Pr[\pbrv(\vec p) = s] = \sum_{A \in {[n] \choose s}} \,\, \prod_{i\in A} p_i \prod_{j\not\in A}(1-p_j),
    \end{equation*}
    where ${[n] \choose s}$ is the collection of all subsets of $[n]$ of size $s$.
    Moreover, the mean  $\E[\pbrv(\vec p)] = \sum_{i=1}^n p_i$ and variance $\sigma^2_{\pbrv(\vec p)} = \sum_{i=1}^n p_i(1-p_i)$.
\end{definition}

Recall that the median of any \gls{rv}\@ $Y$ is a value $m$ such that $\Pr[Y\ge m]\ge \frac12$ and $\Pr[Y\le m]\ge \frac12$. Moreover, the mode of a discrete \gls{rv}\@ is a value that has the maximum point probability. The next few results relate the median and mode to the mean of any \gls{pbd}.

\begin{theorem}[{\cite[Corollary~3.1]{jogdeoMonotoneConvergenceBinomial1968}}] \label{thm:median}
The median of  \(\pbrv(\vec p)\) is either  $\floor{\mu}$ or $\ceil{\mu}$, where   $\mu$ is its mean.
\end{theorem}

\begin{theorem}[{\cite[Theorem~4]{darrochDistributionNumberSuccesses1964}}] \label{thm:mode}
The mode of  \(\pbrv(\vec p)\) is either  $\floor{\mu}$ or $\ceil{\mu}$, where   $\mu$ is its mean.
\end{theorem}

\begin{theorem}[{\cite[Theorem~2]{wang1993number}}]\label{thm:unimodal} 
    The \gls{pbd} is unimodal, i.e.,
    the mode is either unique or shared by two adjacent integers.
    The probability mass decreases monotonically away from the mode.
\end{theorem}

The original papers contain stronger statements that specify the conditions to determine the median/mode.
However, the weaker statement above suffices for the proofs in this paper.

The following theorems bound the loss when using a normal distribution to approximate a \gls{pbd}.
\begin{theorem}[{\cite{shiganovRefinementUpperBound1986}}] \label{thm:normal-apx-to-pbd}
    Let $\mu$ and $\sigma^2$ denote the mean and variance of \(\pbrv(\vec p)\),
    \begin{equation*}
        \max_{0 \leq k \leq n} \abs{\Pr[\pbrv(\vec p)\leq s] - \Phi\left(\frac{k-\mu}{\sigma}\right)}\leq \frac{0.7915}{\sigma},
    \end{equation*}
    where $\Phi$ is the cumulative distribution function of the standard normal distribution.
\end{theorem}

\begin{theorem} [{\cite[Theorem~1.2]{auld2024explicit}}]\label{thm:local-lim}
    Let $\mu$ and $\sigma^2$ denote the mean and variance of \(\pbrv(\vec p)\),

    \[\abs{\Pr[\pbrv(\vec p)  = k] - \frac{1}{\sigma\sqrt{2\pi}}\exp\left(-\frac{(k - \mu)^2}{2 \sigma^2}\right)} \leq \frac{3.23}{\sigma^2} + \frac{1.35}{\sigma^3
    } + \frac{0.25}{\sigma^4}\]
\end{theorem}
A detailed review of the properties of the \gls{pbd} is found in \cite{tangPoissonBinomialDistribution2023}.

\subsubsection{Dynamic Programming Table}\label{sec:intro-dp}
Our analysis crucially relies on visualizing a dynamic programming table, which we now describe.
Fix any non-adaptive solution $\sol$ of testing and assume without loss of generality (by renumbering) that $\sol = \langle1, 2, \dots, n\rangle$.
The  \emph{stages} in the dynamic program correspond to tests: the  $\stg^{th}$ test of $\sol$ is performed in each  stage $\stg\in [n]$. 
We create a dynamic program table $D$ of size $n \times n$.
The entry $D_\stg[j]$ equals the probability of seeing exactly $j$ passes at the end of stage $\stg$, i.e., after having performed tests $\langle1, 2, \dots, \stg\rangle$.

The value of $D_\stg[j]$ can be computed recursively as follows:
For each $\stg \in [n]$ and $j \in [n]$,
\begin{equation*}
    D_\stg[j] =  D_{\stg - 1}[j-1]\cdot p_\stg + D_{\stg-1}[j]\cdot (1-p_\stg),
\end{equation*}
There cannot be more successes than components tested, so $D_\stg[j] = 0$ whenever $j > \stg$.
For the base case, we have $D_\stg[0] = \prod_{i =1}^\stg (1-p_\stg)$ for all $\stg\in [n]$ and $D_1[1] = p_1$.

\paragraph{Non-Completion Probabilities}
From the DP table $D$, we can compute the probability that any particular test  is conducted.
Note that test $\stg+1$ is conducted if and only if  the test outcomes of $\{1,2,\dots ,\stg\}$ are insufficient to evaluate $f$; this corresponds to  the stopping conditions  not being met after stage $\stg$.
Specifically, this means that the number of passes among the first  $\stg$ tests lies in  the interval $N_\stg \coloneq \{(\stg-n+k)^+, \dots, k-1\}$. We  refer to $N_\stg$ as the \emph{non-stopping window} of stage $\stg$. For convenience, we use $N^\ell_\stg := (\stg-n+k)^+$ and $N^h_\stg:= k-1$  to denote the lower and upper bounds in $N_\stg$.
With this definition, we can write
\[\Pr[\text{Test $\stg+1$ conducted}] = \Pr[\pbrv_{\stg}(\vec p)\in N_\stg] = \sum_{j=\stg-n+k}^{k-1} D_\stg[j].\]
Hence, the expected cost $C(\sigma, \vec p)$ from \eqref{eq:cost-fn} can be written  as
\[\sum_{\stg=1}^n c_\stg\cdot \Pr[\pbrv_{\stg-1}(\vec p) \in N_{\stg-1}].\]
See Figure~\ref{fig:dp-visualization} for an illustration of the DP table and the non-stopping window.

\begin{figure}
    \centering
    \tikzmath{\n = 20; \k = 9; \q = 17; \tb = 13;} 

\begin{tikzpicture}

\begin{scope}[yscale=0.4, xscale=0.4]
\fill[fill=green!20] (0, \k-1) rectangle (\n-1, \tb); 

\fill[fill=red!10] ({\n-\k}, 0) -- (\n-1, {\k-1}) -- (\n-1, 0);

\draw (\n-1, 0.3) -- (\n-1, -0.3) node[below]{$n-1$};

\draw[->,thick] (0,0)--(\n,0) node[right]{$\#$ tests};
\draw[->,thick] (0,0)--(0, \tb) node[above]{$\#$ passes};

\draw[thick] (0, \k-1) node[left]{$k-1$} -- (\n-1, \k-1) ;

\draw[thick] ({\n-\k}, 0) -- (\n-1, {\k-1}) node[right]{$k-1$};


\draw[{Bracket[width=0.5em]}-{Bracket[width=0.5em]}] (\q,-\n + \k+ \q) -- node[left]{$N_\nu$} (\q, \k-1);
\draw [decorate,decoration={brace}](\q,1.8) -- node[left]{$E_\nu$} (\q, 4.1);
\pgfmathsetseed{120}

\begin{scope}
    \clip(0,0) rectangle (\n-1,\tb);
\draw [rounded corners] (0,0)
  \foreach \i in {1,...,30} {
    -- (\i, 0.1*\i+0.1*rand)
};
\draw [rounded corners] (0,0)
  \foreach \i in {1,...,30} {
    -- (\i, 0.25*\i+0.1*rand)
};
\end{scope}

\end{scope}
\end{tikzpicture}
    \caption{Visualization of the DP table, and the windows $N_\stg$ and $E_\stg$.
    In the green region, testing stops and we conclude that $f(X) = 1$; in the red region, testing stops and $f(X) = 0$.}
    \label{fig:dp-visualization}
\end{figure}
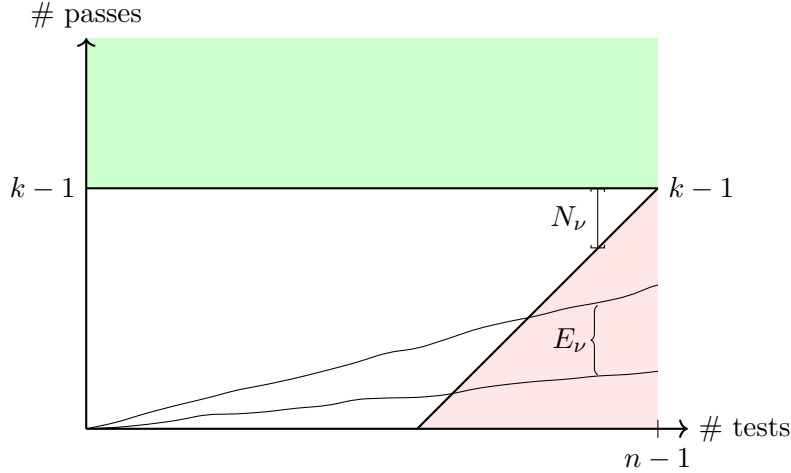
\paragraph{Boundaries for expected number of passes.} We can also view each probability-vector $\vec{p} \in \unset$ as  inducing a path  that traces the expected number of passes at each stage $\stg\in [n]$. 
Formally, path $\Pi(\vec p) \in \mathbb{R}^n$ consists of the values  $\Pi_{\stg}(\vec p) := \sum_{i=1}^\stg p_i$ for $\stg\in [n]$. Our algorithms make use of the ``boundaries'' of such feasible paths. For each stage $\stg$, the lower boundary is given by $E^\ell_{\stg} = \sum_{i=1}^\stg \ell_i$ and the upper boundary is $E_{\stg}^r = \sum_{i=1}^\stg r_i$. We define the {\em expected value window} in stage $\stg$ as $E_{\stg} := [E^\ell_{\stg}, E^h_{\stg}]$. See Figure~\ref{fig:dp-visualization}.

\paragraph{Assumption on $k$.} For the DR \kofn problem we may assume, without loss of generality, that the threshold $k\le \frac{n}2$.   
\begin{proposition} \label{prop:n2-equiv}
    Any distributionally robust $k$-of-$n$ instance with $k > n/2$ has an equivalent \((n-k+1)\)-of-$n$ instance.
\end{proposition}
\begin{proof}
Observe that evaluating $f(X_1,\dots, X_n)=  \indicator{\sum_{i=1}^n X_i \geq k}$ for \(0\)-\(1\) valued $X_i$s is equivalent to evaluating
$\hat{f}(\hat{X}_1,\dots, \hat{X}_n)=  \indicator{\sum_{i=1}^n \hat{X}_i \geq n-k+1}$ where $\hat{X}_i=1-X_i$. Hence, any \kofn instance with probabilities $\{p_i\}_{i=1}^n$ is equivalent to an \((n-k+1)\)-of-$n$ instance with probabilities $\{1-p_i\}_{i=1}^n$.
It now follows that  for any DR \kofn instance  with threshold $k > \frac{n}{2}$ and uncertainty intervals $[\ell_i,r_i]$, we have an equivalent 
    instance  with threshold $n-k+1$ and uncertainty intervals  $[1-r_i , 1- \ell_i]$. 
\end{proof}
Therefore,  we will assume that $k\le n/2$ throughout the paper.

\section{Unit Cost \texorpdfstring{\kofn}{k-of-n}} \label{sec:unit-cost-kofn}
In this section, we show that there is a constant factor approximation for the distributionally robust sequential testing problem when the costs are unit, proving Theorem~\ref{thm:unit-cost}.
The main difficulty in the robust   
\kofn problem lies in the fact that there are two stopping conditions:
\begin{enumerate}[label=(S.\arabic*)]
    \item \label{drst-stop1} at least $k$  tests passed, or 
    \item \label{drst-stop2} at least $n - k + 1$  tests failed. 
\end{enumerate}
These stopping conditions are conflicting: an algorithm that prioritizes \ref{drst-stop1} tends to delay \ref{drst-stop2} and vice versa.
 We will show that under the unit-cost assumption, 
 we can ignore stopping condition \ref{drst-stop2} and still achieve a $2$-approximation algorithm for the \kofn adversary problem.
To this end, we define a new problem \probb with the following stopping conditions:
\begin{enumerate}[label=(S.\Roman*)]
    \item \label{mod-stop1}  at least $k$ tests  passed, or
    \item \label{mod-stop2} at least $ \frac{n}{2}$ tests performed.
\end{enumerate}
For any non-adaptive solution (permutation) $\sol$ and probability distribution $\vec p\in \unset$, the expected testing cost in this new problem  \probb is as follows:
\begin{equation} \label{eq:advbar-cost}
    \bar C(\sol, \vec p) = \sum_{\stg=1}^{n/2} \Pr\left[\pbrv(p_{\sol_1}, p_{\sol_2},\dots, p_{\sol_{\stg-1}}) < k\right].
\end{equation}
Furthermore, let $\advb(\sol) = \max_{\vec p\in \unset}\bar{C}(\sol, \vec p)$ denote the maximum expected cost of solution $\sol$, taken over the uncertainty set $\unset$. Note that this is similar to Definition~\ref{def:adv-defn} (the stopping condition is different). 

The modified problem is easier to analyze, since as we show next, its optimal adversary always  sets $\vec p = \vec \ell$.
The rest of the proof relies on  relating the cost of the new problem \probb  to that of the original problem.

\begin{lemma} \label{lem:robb-bound}
    For every solution $\sol$, the objective for $\probb$ can be bounded as follows:
    \[ \frac12 \adv(\sol) \leq \advb(\sol) \leq \adv(\sol). \]
\end{lemma}
\begin{proof}
The two problems $\adv$ and $\advb$ only differ in their second stopping condition. 

    With the assumption that $ k \leq \frac n2$, stopping condition \ref{drst-stop2} can occur only after we have performed at least $\frac n2$ tests.
    However, in \probb, we stop as soon as  we have performed  $\frac n2$ tests.
    In other words, stopping condition \ref{mod-stop2} occurs no later than \ref{drst-stop2} under any realization. So, $\advb(\sol) \leq \adv(\sol)$.

    To prove the lower bound on $\advb(\sol)$, fix any outcome $x \in \{0, 1\}^n$ of the tests.
    If there are $k$ passes from the first $\frac n2$ tests according to permutation $\sol$, then $\adv(\sol) = \advb(\sol)$.
    On the other hand, if the outcome $x$ terminates after $\frac n2$ tests in the original \gls{adv} problem, it would have a cost between $\frac{n}{2}$ and $n$ (since all costs are unit).
    This will have cost exactly $\frac n2$ in \probb.
    Thus, the lower bound follows.
\end{proof}

\begin{lemma} \label{lem:unit-cost-adv}
    For the modified problem \(\probb\), regardless of the solution $\sol$, it is optimal to set
    $\vec p = \vec \ell$
\end{lemma}
\begin{proof}
    Fix any solution $\sol$ and assume without loss of generality that $\sol = \langle 1, 2, \dots, n\rangle.$
    Moreover, fix any vector $\vec p \in \unset$ with some index $j$ where $p_j > \ell_j$.
    We want to show that lowering $p_j$ will only increase the expected cost $\bar{C}$.
    The proof analyzes a DP table, denoted by $T$, for computing $\bar C(\sol, \vec p)$.
    The value $T_{\stg}[j]$ gives the expected cost of testing the suffix $\langle \nu, \nu+1, \dots, n\rangle$ given that $j\in \{0, 1,\dots, n\}$ passed tests were observed  in the prefix $\langle 1, 2, \dots, \stg - 1\rangle$.
    For every stage $\stg \in [\frac n2]$ and number of passes $ j \in \{0, 1, \dots, n-1\}$ the recurrence is given by 
    \begin{equation} \label{eq:cbar-recurs}
        T_{\stg}[j] = 
        \left\{ \begin{array}{ll}
            0 & \mbox{ if $\nu \geq \frac{n}{2}$ or $j \geq k$} \\
                    1 + T_{\stg+1}[j](1-p_{\nu}) + T_{\nu + 1}[j+1] p_\stg & \mbox{ otherwise} 
        \end{array}
        \right. .        
     \end{equation}
The first case above ($\nu \geq \frac{n}{2}$ or $j \geq k$) corresponds to the stopping conditions.
    The objective value $\bar{C}(\asol, \vec p)$ is stored in $T_{1}[0]$.
    First, we claim that 
    \begin{quote}
For each $\nu\in [n]$, $T_\nu[j]$ is non-increasing in $j$.         
    \end{quote}
    
    We prove this claim by induction on the stage $\nu=n,n-1,\dots, 1$.
    For the base case, where $\nu \geq \frac n 2$, all entries are $0$, so it is non-increasing in $j$.
    For the inductive step, assume that $T_{\stg+1}[j]$ is non-increasing in $j$ for some stage $\stg< n/2$. Clearly, $T_\stg[j]$ is non-increasing for $j\ge k$ (all these values are zero). 
    For any $j <k$,  
    \begin{align*}
        T_{\stg}[j] &= 1 + T_{\stg+1}[j]\cdot (1-p_{\stg}) + T_{\stg + 1}[j+1] \cdot p_\stg  \\
        &\ge1 + T_{\stg+1}[j+1] \cdot (1-p_{\stg}) + T_{\stg + 1}[j+2] \cdot p_\stg  \,\,\ge\,\, T_{\stg}[j+1],
    \end{align*}
    where the first inequality is by induction. The last inequality is by the recurrence \eqref{eq:cbar-recurs}.
    
    Fix any stage $\hat\stg$ where $p_{\hat \stg} > \ell_{\hat \stg}$. 
    Consider $\vec p'$, obtained by lowering $p_{\hat \stg}$ to $\ell_{\hat \stg}$. We will show that $\bar{C}(\sol, \vec p)\le \bar{C}(\sol, \vec p')$. If $\hat\stg\ge n/2$ then   $\bar{C}(\sol, \vec p) = \bar{C}(\sol, \vec p')$; so we assume  $\hat\stg < n/2$ below.     
    Let $U$ be the new DP table for computing the expected cost $C(\sol, \vec p')$; we will show that $U_\stg[j] \ge T_\stg[j]$ for all $\stg$ and $j$. Note that $U_\stg[j] = T_\stg[j]$ for all $\stg > \hat\stg$ as all probabilities in stages  after $\hat \stg$ are equal. Also, $U_\stg[j] = T_\stg[j]=0$ for $j\ge k$. For any $j<k$, we have 
\begin{align*}
{U}_{\hat{\stg}}[j] & = 1 + T_{\hat \stg + 1}[j] \cdot (1-\ell_{\hat \stg}) + T_{\hat\stg + 1}[j+1] \cdot \ell_{\hat\stg} = 1+  T_{\hat\stg + 1}[j] +  \ell_{\hat\stg} \cdot ( T_{\hat\stg + 1}[j+1] - T_{\hat\stg + 1}[j])  \\
& \geq  1+  T_{\hat\stg + 1}[j] +  p_{\hat\stg} \cdot ( T_{\hat\stg + 1}[j+1] - T_{\hat\stg + 1}[j])
= 
1 + T_{\hat \stg + 1}[j] \cdot (1-p_{\hat\stg}) + T_{\hat \stg + 1}[j+1] \cdot p_{\hat\stg}  = T_{\hat\stg}[j].
     \end{align*}
     The first equality uses $U_{\hat \stg + 1}=T_{\hat \stg + 1}$, and the inequality   uses  $T_{\hat\stg + 1}[j+1] \le T_{\hat\stg + 1}[j]$ and $p_{\hat \stg} > \ell_{\hat \stg}$. So, we have $U_{\hat \stg} [j]  \ge T_{\hat \stg} [j]$ for all  $j$. For any stage $\stg< \hat \stg$, the recurrence for both $T_\stg$ and $U_\stg$ are the same~\eqref{eq:cbar-recurs}. Using the fact that all coefficients in this recurrence are non-negative, it now follows that $U_{\hat \stg} [j]  \ge T_{\hat \stg} [j]$ for all $\stg$ and  $j$.
    Hence, $\bar C(\sol, \vec p') = U_1[0] \geq T_1[0] = \bar C(\sol, \vec p)$. 

    Finally, we repeat the above argument  for each stage $\hat\stg\in [n]$ to obtain $\bar C(\sol, \vec 
    \ell)\ge \bar C(\sol, \vec p)$ for  any $\vec p\in \unset$.   
\end{proof}

Combining Lemmas~\ref{lem:robb-bound} and~\ref{lem:unit-cost-adv}, we obtain a $2$-approximation algorithm for the unit-cost adversary problem. We now use this to obtain an approximation algorithm for  the DR \kofn problem.

\begin{lemma} \label{lem:unit-cost-policy}
    The solution $\asol$ that minimizes $\advb$ sorts by decreasing $\ell_i$.
\end{lemma}
\begin{proof}
    We can prove this claim by a swapping argument.
    Let $\asol$ be the solution that sorts by decreasing $\ell_i$, and assume that $\asol = \langle 1, 2, \dots, n\rangle.$
    This implies that $\ell_1 \geq \ell_2 \geq\dots\geq \ell_n$.
    Suppose that we have a solution $\asol'$ where some test $t$ is conducted immediately before test $s$, but $s<t$, i.e., $\asol' =\langle I, t, s, O\rangle$, where $I$ is the prefix of tests conducted before $t$, and $O$ is the suffix of tests conducted after $s$.
    We want to show that the solution $\asol = \langle I, s, t, O\rangle$ that swaps $s$ and $t$ is no worse than $\asol'$.

    Using Lemma~\ref{lem:unit-cost-adv}, we know that $\vec p = \vec \ell$ maximizes $\bar{C}(\sol,\vec p)$ for all solutions $\sol$. 
    Let $X_i\sim \bern(\ell_i)$ for all $i\in [n]$. Then, using \eqref{eq:advbar-cost}, the worst-case expected cost
    \begin{equation} \label{eq:adv-cost-expanded}
        \advb(\sol) = \bar C(\sol, \vec \ell) = \sum_{\stg=1}^{n/2} \Pr\left[\sum_{i=1}^{\stg-1} X_{\sol_i}< k\right], \quad\text{for all $\sol$}.
    \end{equation}
    We use this expanded definition \eqref{eq:adv-cost-expanded} to compare the difference between $\bar C(\asol, \vec \ell)$ and $\bar C(\asol', \vec \ell)$.
    In particular, we compare them 
    term by term. For $\stg \le \abs I + 1$, the probability terms in \eqref{eq:adv-cost-expanded} are the same, since the $\stg^{th}$ prefix of $\asol$ and $\asol'$ are equal.
    Moreover, for $\stg \geq \abs I + 3$, the terms within the summation $\sum_{i=1}^{\stg-1}X_{\sol_i}$ are the same despite swapping $s$ and $t$.
    So, the only difference between $\bar C(\asol, \vec \ell)$ and $\bar C(\asol, \vec \ell)$ is the $(I + 2)$'th term, where the probability term is $\Pr[\sum_{i\in I}  X_{i} + X_s < k]$ for $\asol$ but $\Pr[\sum_{i\in I}  X_{i} + X_t<k]$ for $\asol'$.
     Thus,
    \begin{align*}
        \advb(\asol') - \advb(\asol) &= \bar C(\asol', \vec\ell) - \bar C(\asol, \vec \ell) \\
        &= \Pr\left[\sum_{i\in I} X_{i} + X_t < k \right] -\Pr\left[\sum_{i\in  I} X_{i} + X_s < k \right]\,\, \geq\,\, 0,
    \end{align*}
    where the inequality uses $\ell_t \le \ell_s$.
\end{proof}

Combining Lemmas~\ref{lem:robb-bound} and~\ref{lem:unit-cost-policy} we obtain a $2$-approximation algorithm for the DR \kofn problem, completing the proof of Theorem~\ref{thm:unit-cost}. We note that the algorithm  discussed above assumes that $k\le n/2$, which by Proposition~\ref{prop:n2-equiv} is without loss of generality. If we combine the ``reduction'' in  Proposition~\ref{prop:n2-equiv} with the above approach, the algorithm  for the case $k>n/2$ just sorts by increasing $r_i$.

\section{General Cost \texorpdfstring{\kofn}{k-of-n}} \label{sec:general-cost-kofn}

We now consider the setting with general costs and prove Theorem~\ref{thm:gen-cost}.  This is much more complex than the unit-cost case because we can no longer  eliminate one stopping condition. In fact, our approximation ratio for the general problem depends on a ``boundedness''  parameter of the instance: see Definition~\ref{def:eps-bounded}. Recall that in an $\epsilon$-bounded instance, all the uncertainty intervals are contained in $[\epsilon,1-\epsilon]$. Our main result is an $O(1/\sqrt{\epsilon})$-approximation algorithm for both the adversary and DR \kofn problems.    

A natural approach for the general cost problem is to extend the greedy algorithm from the unit-cost case. In Appendix~\ref{app:extension-to-gen-cost} we show that the approximation ratio of such greedy algorithms is at least $\Omega(1/\epsilon)$. 

Instead, our approach for the general cost setting is as follows. 
First, we (approximately) maximize the stage-wise non-completion probabilities. That is, for each stage $\stg \in [n]$, we find a vector $\vec p \in \unset$   that maximizes $\Pr[\pbrv_{\stg}(\vec p) \in N_\stg]$. We obtain an $O(1/\sqrt{\epsilon})$-approximation algorithm  that involves  two cases, depending on whether/not the expected-value window $E_\stg$ overlaps with the non-stopping $N_\stg$.   The analysis here relies on various properties of PBDs (such as the unimodal property and normal approximations) as well as an anti-concentration property for $\epsilon$-bounded instances. Next, we show that these non-completion probabilities across stages can  be (approximately) maximized simultaneously (i.e., the same probability-vector $\vec p$ works for all stages), leading to the approximation algorithm for the adversary problem. The main idea here is to look at the relative positions of the  expected-value and non-stopping windows in the last stage.  Finally, we use this (approximate) characterization of the adversary to obtain the 
approximation algorithm for the DR problem.

\subsection{Maximizing  Non-Completion Probabilities at a Single Stage} \label{sec:single-stg-apx}

In this section, we assume without loss of generality that $\sol = \langle 1, 2, \dots, n\rangle.$
Recall that we visualize running $\sol$ on a realization as a path along the DP table (Section~\ref{sec:intro-dp}).
Moreover, when tests $\langle 1, 2, \dots, \nu\rangle$ are conducted, the term $\pbrv_\nu(\vec p)$ is the \gls{rv}\@ representing the number of tests passed (i.e., $\pbrv_\nu (\vec p) = \sum_{i=1}^\stg X_i$).
When the number of passes lies within the non-stopping window $N_\nu = \{(\nu - n+k)^+,\dots,  k-1\}$, testing is not complete, and the $(\nu+1)$th test has to be conducted.

In this section, we maximize the non-completion probabilities for a fixed stage $\stg \in [n]$.
i.e.,\@
\begin{equation} \label{eq:stagewise-noncompletion}
    \max_{\substack{\vec p \in \unset}} \Pr[\pbrv_\stg(\vec p) \in N_\stg].
\end{equation}
In particular, we show the following lemma:
\begin{lemma} \label{lem:stagewise-noncomp}
    At every stage $\stg \in [n]$, there is a $\bigOh{1/\sqrt\epsilon}$-approximation for maximizing the non-completion probability $\Pr[\pbrv_{\stg}(\vec p)  \in N_{\stg}]$.
\end{lemma}

Although the non-stopping window $N_\stg$ is  a discrete set, we treat it as a continuous interval for our analysis.
Moreover, for the analysis, we widen this interval at stages $\stg \geq \frac n2$, as follows:
\begin{definition}
    The modified non-stopping window, $\tn$, is equal $N$ at stages $\stg < \frac n2$ and extended by one (i.e., $[N^\ell-1, N^h+1]$) when $\stg \geq \frac n2$ :
    \begin{align*}
    \tn= \begin{cases}
        [0, k-1] & \text{if } \stg < \frac n2,\\
        [(\stg-n+k-1)^+, k] & \text{otherwise.}
    \end{cases}
    \end{align*}
\end{definition}

As the results in this subsection focus on a fixed stage $\stg$, we often drop the subscript  $\stg$ to reduce clutter.
Recall that the window $E= [\sum_{i=1}^\stg \ell_i, \sum_{i=1}^\stg r_i]$ gives the range of  expected number of passes among $\stg$ tests, and $N$ gives the non-stopping window.

The results in this section depend on the relative position of $E$ and $\tn$.
When $E\cap \tn \neq \emptyset$, we say that the windows \emph{overlap}; otherwise, we say that they are \emph{disjoint}.
If $E$ and $\tn$ are disjoint, we say that $E$ is \emph{above} $\tn$ when $E^\ell > \tn^h$, and that $E$ is \emph{below} $\tn$ if $E^h < \tn^\ell$.
When the windows overlap, we can pick some $\vec p\in \unset$ such that $\vec p \in \tn$.
We will show that such a setting of $\vec p$ ensures that the probability that $\pbrv(\vec p) \in N$ is significant, and thus guarantees an $O(1/\sqrt\epsilon)$ approximation for \eqref{eq:stagewise-noncompletion}.
On the other hand, when the windows are disjoint, then we would want to pick some $\vec p$ such that the expectation $\sum_{i=1}^\stg p_i$ is as close to the non-stopping window as possible.
This turns out to be optimal for \eqref{eq:stagewise-noncompletion}.

The analysis in this section will differ based on whether $\stg < \frac n2$ or $\stg \geq \frac n2$, even though the results are very similar.
Because of our assumption that $k \leq \frac n2$, the non-stopping window $N_\stg$ cannot be exited from below when $\stg < \frac n2$ since even seeing $\stg = (\frac n2-1)$ failed tests is inconclusive.
Exploiting this property is necessary at small stages because concentration inequalities are weak in the regime where $\stg \ll \frac n2$.
On the other hand, when $\stg \geq \frac n2$, we might exit the non-stopping window from below or above (i.e., testing stops because we see too many fails or too many passes). Thus, we need to use concentration inequalities to bound the tail probabilities.

\begin{figure}
    \centering
    \tikzmath{\n = 20; \k = 9; \q = 15; \tb = 13;} 
    \tikzmath{\n = 20; \k = 9; \q = 15; \tb = 13;} 
\begin{tikzpicture}
\begin{scope}[yscale=0.4, xscale=0.4]

\fill[fill=green!20] (0, \k-1) rectangle (\n/2, \tb); 
\fill[fill=green!20] (0, \k) rectangle (\n, \tb); 

\fill[fill=red!10] ({\n-\k}, 0) -- (\n, {\k-1}) -- (\n, 0);

\draw[->,thick] (0,0)--(\n+0.5,0) node[right]{$\#$ tests};
\draw[->,thick] (0,0)--(0, \tb+0.5) node[above]{$\#$ passes};
\draw (\n, 0.3) -- (\n, -0.3) node[below]{$n$};

\draw[dotted] (\n/2, 0) -- (\n/2, \tb);

\draw (0, \k-1) node[left]{$k-1$} -- (\n/2, \k-1);
\draw (\n/2, \k) -- (\n, \k) node[right]{$k$};

\draw ({\n-\k}, 0) -- (\n, {\k-1}) node[right]{$k-1$};

\node at (\n/2, 11) {\textbf{Lemma~\ref{lem:disjoint-optimal-above}}};
\node at (\n/2, 5) {\textbf{Corollary~\ref{cor:overlap-prob-bound}}};
\node at (\n/2+6.3, 1.5) {\textbf{Lemma~\ref{lem:disjoint-optimal-below}}};

\end{scope}
\end{tikzpicture}
    \caption{Summary of Cases and their corresponding results.}
    \label{fig:cases-signpost}
\end{figure}
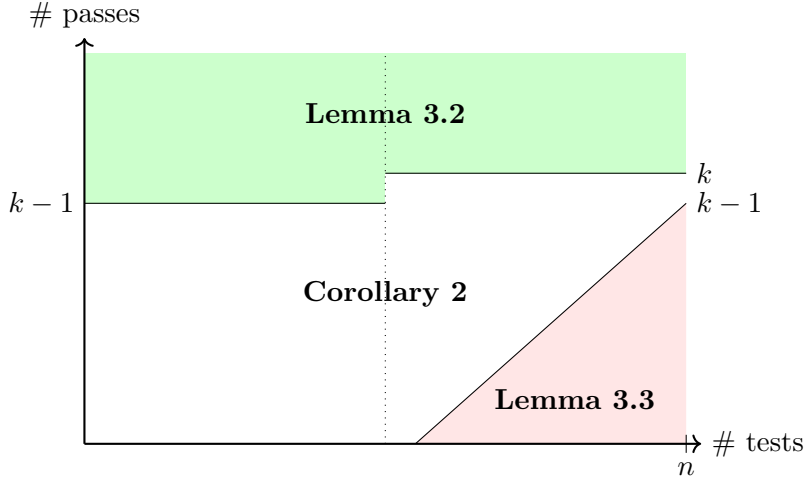

\paragraph{Disjoint Windows}
When the windows are disjoint, we have an optimal setting of $p_i$.
The case where $\stg < \frac n2$ is easier, since we can use the property that $\tn^\ell = 0$.
On the other hand, when $\stg \geq \frac n2$, we no longer have the guarantee that $N^\ell = 0$.
In fact, at stage $n-1$, we have $N^\ell = N^h$.
Thus, we need a stronger assumption that $E^\ell \geq N^h + 1$.

\begin{lemma} \label{lem:disjoint-optimal-above}
    If $E$ is above $\tn$, it is optimal to set \(\vec p = \vec \ell\).
\end{lemma}
\begin{proof}
    Fix any window $E$ that $E^\ell > \tn^h$.
    Suppose that there exists some index $j$ such that $p_j > \ell_j$.
    We want to show that we can lower $p_j$ to $\ell_j$ to get $\vec p'$ such that 
    \(\Pr[\pbrv(\vec p')\in N] \geq \Pr[\pbrv(\vec p) \in N].\)
    We define $\pbrv_{-j} = \sum_{i\in [n]: i \neq j} X_i$, so $\pbrv_{-j}$ follows a \gls{pbd} of order $\stg-1$ (i.e., $\pbrv_{-j} \sim \pobin(p_1, \dots, p_{j-1},p_{j+1}, \dots, p_\stg)$).
    We can thus write the \gls{rv}\@ $\pbrv'$, obtained by lowering $p_j$ to $\ell_j$, as $\pbrv'=\pbrv_{-j} + X'$ where $X'\sim\bern(\ell_j)$.
    We  decompose the event $\pbrv\in \{N^\ell, \dots, N^h\}$ as

     \begin{equation} \label{eq:event-decomp}
         ((\pbrv_{-j}=N^\ell - 1) \land (X_j  =1)) \lor (\pbrv_{-j} \in \{N^\ell, \dots, N^h-1\}) \lor ((\pbrv_{-j} = N^h) \land X_j = 0).
     \end{equation}
    A similar decomposition holds for the event $\pbrv' \in \{N^\ell, \dots, N^h\}$, where we replace all mentions of $X_j$ in \eqref{eq:event-decomp} with $X_j'$.
    As a consequence, we have
    \begin{align}
        &\Pr[\pbrv' \in \{N^\ell,\dots,  N^h\}] - \Pr[\pbrv\in \{N^\ell, \dots, N^h\}] \notag\\
        &\qquad =\Pr[\pbrv_{-j} = N^\ell - 1]\Pr[X_j' =1] + \Pr[\pbrv_{-j} \in \{N^\ell,\dots, N^h-1\}] + \Pr[\pbrv_{-j} = N^h]\Pr[X_j' = 0] \notag\\
            &\qquad\quad -\left(\Pr[\pbrv_{-j} = N^\ell - 1]\Pr[X_j =1] + \Pr[\pbrv_{-j} \in \{N^\ell,\dots, N^h-1\}] + \Pr[\pbrv_{-j} = N^h]\Pr[X_j = 0]\right) \notag\\
        &\qquad = \Pr[\pbrv_{-j} = N^\ell - 1]\ell_j + \Pr[\pbrv_{-j} = N^h](1-\ell_j) - \Pr[\pbrv_{-j} = N^\ell - 1]p_j - \Pr[\pbrv_{-j} = N^h](1-p_j) \notag\\
        &\qquad= (\Pr[\pbrv_{-j} = N^h] -\Pr[\pbrv_{-j} = N^\ell - 1])(p_j - \ell_j) \notag\\
        &\qquad \geq 0 \label{eq:disjoint-geq},
    \end{align}

    We now prove \eqref{eq:disjoint-geq} by cases:
    \begin{itemize}
        \item Case $N^\ell = 0$: Under this case, $\Pr[\pbrv_{-j} = N^\ell - 1] = 0$ since the support of a \gls{pbd} is non-negative.
        In this case, the LHS is clearly non-negative.
        Moreover, by the definition of $N$, all stages $\stg < \frac n2$ fall within this case.
        \item Case $N^\ell > 0$:
        This case can occur only at stages $\stg \ge \frac n2$. By the definition of $\tn$, the condition that $E$ is above $\tn$ is equivalent to $E^\ell > N^h+1$.This implies $\pbrv_{-j}$ has mean greater than $N^h$, which by Theorem~\ref{thm:mode} means that the mode of $\pbrv_{-j}$ is at least $N^h$.
        Since the probability mass of a \gls{pbd} decreases monotonically from the mode (Theorem~\ref{thm:unimodal}), we have \(\Pr[\pbrv_{-j} = N^h] \geq \Pr[\pbrv_{-j} = N^\ell - 1]\).
    \end{itemize}
\end{proof}
Symmetrically, we have:
\begin{lemma} \label{lem:disjoint-optimal-below}
    If \(E\) is below \(\tn\), it is optimal to set \(\vec p = \vec r\).
\end{lemma}
\begin{proof}
    We follow the same structure as Lemma~\ref{lem:drst-disjoint} to show that we would want to raise $p_j$ to $r_j$.
    The same decomposition in \eqref{eq:event-decomp} holds.
    This gives us:
    \begin{align*}
        \Pr[\pbrv' \in \{N^\ell, \dots, N^h\}] - \Pr[\pbrv \in \{N^\ell, \dots, N^h\}] &= (\Pr[\pbrv_{-j} \in N^\ell - 1] - \Pr[\pbrv_{-j} = N^h])(r_j - p_j)\\
        &\geq 0.
    \end{align*}

    The difference in this proof is in showing the last inequality.
    Here, note that $E^h < N^\ell - 1$, so the mode of \(\pbrv_{-j}\) is at most $N^\ell - 1$.
    Then Theorem~\ref{thm:unimodal} tells us that the probability mass decreases monotonically away from the mode, so $\Pr[\pbrv_{-j} \in N^\ell - 1] \ge \Pr[\pbrv_{-j} = N^h]$. 
\end{proof}

\begin{corollary}
    When $E$ and $\tn$ are disjoint.
    There is an optimal setting of \(\vec p \in \unset\) that maximizes the non-stopping probability.
    Namely, if $E$ is above $\tn$, then we want to set \(\vec p = \vec \ell\), otherwise, if $E$ is below $\tn$, then we want to set \(\vec p = \vec r\).
\end{corollary}

\paragraph{Overlapping Windows}
When the windows overlap, we can set $\vec p \in \unset$ such that $\sum_{i=1}^\stg p_i \in \tn$.
We lower bound the non-completion probability by any such setting of $\vec p$. We then compare the lower bound with an upper bound on the non-completion probability achievable by any valid setting of $\vec p \in \unset$.
Let $w_\stg = \abs{N} = N^h - N^\ell + 1=\min(n-\stg,k-1)$ be the cardinality of the non-stopping window $N$ when $\stg$ components are tested.
Since $k$ and $n$ are specified by the problem instance, $w$ is a function of $\stg$ only.
We first state the lower and upper bounds on $\Pr[\pbrv(\vec p) \in N]$.
The rest of this section is used to prove these bounds.

\begin{definition} \label{def:alpha-beta-bounds}
    We define the lower bound $\alpha_\stg$ and upper bound $\beta_\stg$ as 
    \[\alpha_{\stg} = \begin{cases}
        \frac 14 & \text{if } 3\sqrt\stg \leq \ceil{\frac w2}\\
        \frac{w/4}{6\sqrt \stg}&\text{otherwise}\\
    \end{cases} \qquad \text{and} \qquad 
    \beta_{\stg} =     \min\left\{\frac{2\sqrt2 \, w}{\sqrt{\epsilon \stg}}, 1\right\},\]
where $w=\min(n-\stg, k-1)$. 
\end{definition}

\begin{lemma} \label{lem:alpha-lb}
    Any probability vector $\vec p\in\unset$ such that $\sum_{i=1}^\stg p_i\in \tn$ satisfies the lower bound
    \[ \Pr[\pbrv(\vec p) \in N] \geq \alpha_\stg.\]
\end{lemma}

First, consider the case where $\stg < \frac n2$.
\begin{lemma} \label{lem:small-stage-lb}
    Fix any stage $\stg < \frac n2$.
    Any setting of $\vec p \in \unset$ with $\mu = \sum_{i=1}^\stg p_i \in \tn$ will result in
    \[\Pr[\pbrv(\vec p) \in N] \geq \frac 12 \geq \alpha_\stg.\]
\end{lemma}

\begin{proof}
    Fix any setting of ${\vec p} \in \unset$ with $\mu =  \sum_{i=1}^\stg  p_i \in N$ (i.e., $\mu \leq k-1$).
    Since the median is at most $\ceil{\mu}$ by Theorem~\ref{thm:median}, we have $\Pr[\pbrv({\vec p}) \leq \ceil{\mu}] \geq 0.5$.
    Note that the support of a \gls{pbd} is non-negative so $\Pr[0 \leq \pbrv(\vec p) \leq \ceil\mu] \ge 0.5 $.
    Moreover, since $k-1$ is an integer, the fact that $ \mu\leq k-1$ implies that $\ceil \mu \leq k-1$. Hence, the interval $[0, \ceil\mu] \subseteq N$ and 
    \(\Pr[{\pbrv}({\vec p}) \in N] \geq \Pr[0 \leq \pbrv({\vec p}) \leq \ceil \mu] \ge 0.5.\)
\end{proof}

Before proving the lower bound for the case where $\stg \geq n/2$, we first establish some bounds on the standard deviation.
\begin{lemma} \label{lem:var-bound}
    For any stage \(\stg \in[n]\) and for any \(\vec p \in \unset\), we have
    \[ \sqrt{\frac{\epsilon\nu}2} \leq \sigma_{\pbrv(\vec p)} \leq \frac{\sqrt\nu}{2}.\]
\end{lemma}
\begin{proof}
    The variance $\sigma^2_{\pbrv(\vec p)} = \sum_{i=1}^\nu p_i (1-p_i)$.
    The variance is minimized when $p_i \in \{\epsilon, 1-\epsilon\}$ for all $i$, so $\sigma_{\pbrv(\vec p)} \geq \sqrt{\nu\epsilon (1-\epsilon)} \geq \sqrt{\frac {\epsilon\stg} 2}$ for $\epsilon \in [0, \frac 12]$.
    On the other hand, the variance is maximized when $p_i = \frac 12$ for all $i$, so $\sigma_{\pbrv(\vec p)} \leq\frac{\sqrt\nu}{2}$.
\end{proof}

\begin{lemma}\label{lem:rhs-conc}
    For every $\vec p \in \unset$.
    Letting $\mu = \sum_{i=1}^\stg p_i$,
    we have 
    \[ \Pr[\floor \mu \leq \pbrv(\vec p) \le \mu + 3 \sqrt \nu ] \geq 0.49 \quad \text{and} \quad     \Pr[\mu - 3 \sqrt \stg \leq \pbrv(\vec p) \leq \ceil\mu]\geq 0.49.\]
\end{lemma}
\begin{proof}
    We show only the first inequality \( \Pr[\floor{\mu} \leq \pbrv(\vec p) \leq \mu + 3 \sqrt\stg] \geq 0.49\).
    The proof of the other inequality is symmetrical.
    Fix any $\vec p \in \unset$.
    Using Hoeffding's inequality, we get
    \[\Pr[\pbrv(\vec p) \geq \mu + 3 \sqrt \stg] \leq\exp\left(-\frac{2(3\sqrt \stg)^2}{\stg}\right) = e^{-18} \leq 0.01.\]
    The median of $\pbrv(\vec p)$ is at least $\floor \mu$ by Theorem~\ref{thm:median}.
    Thus, $\Pr[\pbrv(\vec p) \geq \floor \mu] \geq 0.5$.
    The lemma follows from the relation
    \[\Pr[\floor \mu \leq \pbrv(\vec p) \leq \mu + 3 \sqrt \stg] = \Pr[\pbrv(\vec p) \geq \floor \mu] - \Pr[\pbrv(\vec p) \geq \mu + 3 \sqrt \stg] \ge 0.5 - 0.01 = 0.49.\]
\end{proof}

We show that when the probabilities are bounded away from $0$ and $1$, we have a constant bound on the probability mass at the mode.
\begin{lemma} \label{lem:mass-ub}
    Suppose that \(\pbrv(\vec p)\) has \(\vec p \in [\epsilon, 1-\epsilon]^\stg\) and order \(\stg \geq \frac{50}{\epsilon}\), then \(\Pr[\pbrv(\vec p) = m] \leq 0.24\) for every $m\in [\stg]$.
\end{lemma}
\begin{proof}
    Fix any $\vec p \in [\epsilon, 1-\epsilon]^\stg$.
    Define $\mu \coloneq\sum_{i=1}^\stg p_i$ be the mean of $\pbrv(\vec p)$ and let $\sigma$ denote the standard deviation of $\pbrv(\vec p)$.
    We define the normal \gls{rv}\@ $Z\sim \normaldist(\mu, \sigma^2)$.
    The probability density function of
    $Z$ is $f(x) = \frac{1}{\sigma\sqrt{2\pi}}\exp\left(-\frac{(x-\mu)^2}{2\sigma^2}\right)$.
    The maximum density is achieved when $x = \mu$ with density $f(\mu) = \frac{1}{\sigma\sqrt{2\pi}}$.
    Using Theorem~\ref{thm:local-lim}, we can bound the difference between the density of the normal pdf and the pmf of $\pbrv(\vec p)$.
    Thus, for every $m\in [\stg]$, we have
    \[\Pr[\pbrv(\vec p)= m] \leq \frac{1}{\sigma \sqrt{2\pi}} + \frac{3.23}{\sigma^2} + \frac{1.35}{\sigma^3} + \frac{0.25}{\sigma^4}\leq 0.24\]
    whenever $\sigma \geq 5$.
    Since $\sigma \geq \sqrt{\frac{\epsilon \stg}{2}}$, it suffices to have $\stg \geq \frac{50}{\epsilon}$.
\end{proof}

The following lemma is similar to Lemma~\ref{lem:small-stage-lb}, but they differ in the property used.
Lemma~\ref{lem:small-stage-lb} used the property that $N^\ell = 0$; the next lemma uses a concentration inequality, which is permissible since $\stg \geq \frac n2$.

\begin{lemma} \label{lem:large-stg-lb}
    When $\stg \geq \frac n2$, any setting of $\vec p \in \unset$ with $\sum_{i=1}^\stg p_i \in \tilde N$ will result in
    \[\Pr[\pbrv(\vec p) \in N] \geq \alpha_\stg.\]
\end{lemma}
\begin{proof}
    Fix any setting of $\pbrv({\vec p})$ such that $ \mu = \sum_{i=1}^\stg  p_i \in \tilde N$.
    Under the assumption that $\stg > \frac n2$, $\tn = [N^\ell-1, N^h +1]$.
    We prove the lemma under the cases $\mu \in [N^\ell - 1, N^\ell)$, $\mu \in [N^\ell, N^h]$ or $\mu \in (N^h, N^h + 1]$.
    We first show the case where $\mu \in [N^\ell, N^h]$.

    When $\mu \in N$, at least one of $[\mu, N^h]$ or $[N^\ell, \mu] $ contain at least $\ceil{w/2}$ integer points.
    Without loss of generality, assume that $ [\mu, N^h]$ has at least $\ceil{w/2}$ integer points.
    Since $ N^\ell $ is an integer, $ \mu \geq N^\ell $ implies that $ \floor \mu \geq N^\ell$, so $ [\floor\mu,  \mu + \frac w2 ]\subseteq [N^\ell, N^h]$.
    If $\ceil {\frac w2} > 3 \sqrt\stg$, then we are done, since we now have $[\floor \mu, \mu + 3\sqrt \stg] \subseteq[\floor \mu, \mu + \frac w2]$ and
    \[\Pr[\pbrv(\stg) \in N] \geq \Pr\left[\pbrv(\vec p) \in \left[\floor\mu, \mu + \frac w2\right]\right] \geq \Pr[\pbrv(\vec p) \in [\floor \mu , \mu+ 3\sqrt\stg]] \geq 0.49,\]
    where the last inequality follows from Lemma~\ref{lem:rhs-conc}.

    On the other hand, suppose that $\ceil {\frac w2} \leq 3 \sqrt \stg$.
    Then we need to use the property that the mode is either $\floor \mu$ or $\ceil \mu$ and that the \gls{pbd} is a unimodal distribution where the probability mass decreases monotonically away from the mode (Theorem~\ref{thm:unimodal}).
    These properties imply that the interval $[\floor\mu, \mu + \frac w2]$ captures the highest $\ceil{\frac w2}$ probability mass within the interval $[\floor{\mu}, \mu + 3\sqrt\stg]$.
    Thus,
    \[\Pr[\pbrv(\vec p) \in N] \geq \Pr\left[\pbrv(\vec p)\in \left[\floor \mu, \mu + \frac w 2\right]\right]\geq \frac{w/2}{3\sqrt \stg} \Pr[\pbrv(\vec p) \in [\floor \mu, \mu + 3\sqrt\stg] \geq \frac{0.49(w/2)}{3\sqrt\stg}.\]

    We now consider the case where $\mu \in [N^\ell -1, N^\ell)$.
    The analysis is largely similar, but we need to account for the fact that the mode of $\pbrv(\vec p)$ might be on $N^\ell$.
    As $\stg \geq \frac n2 > \frac {50}\epsilon$, we can apply Lemma~\ref{lem:mass-ub} to get 
    \begin{align*}
        \Pr[N^\ell \leq \pbrv(\vec p) \leq \mu + 3\sqrt\stg] &= \Pr[\ceil\mu \leq \pbrv(\vec p) \leq \mu + 3\sqrt \stg] \\
        &= \Pr[\floor \mu \leq \pbrv(\vec p) \le \mu + 3 \sqrt\stg] - \Pr[\pbrv(\vec p) = \floor\mu]\\
        &\ge 0.49 - 0.24 = 0.25.
    \end{align*}
    If $\ceil {w /2}\geq 3 \sqrt \stg$, then we are done, since that would mean $[N^\ell, N^h] \supseteq [N^\ell, \mu + 3\sqrt \stg]$ and thus,
    \[\Pr[\pbrv(\vec p)\in N] \geq \Pr[N^\ell \leq \pbrv(\vec p) \leq \mu + 3\sqrt\stg] \geq 0.25.\]
    On the other hand, suppose $\ceil {\frac w2} \leq 3 \sqrt\stg$.
    Again, the same argument that the mode is at most $\ceil \mu$ and the \gls{pbd} decreases monotonically away from the mode implies that $[N^\ell, N^h]$ is capturing the top $\frac w2$ probability mass among the $3\sqrt\stg$ points.
    By an averaging argument, we get
    \[\Pr[\pbrv(\vec p) \in N] \geq \frac{0.25 (w/2)}{3 \sqrt \stg}.\]
    The case where $\mu \in (N^h, N^{h+1}]$ is symmetrical.
\end{proof}

We have now proved that $\alpha_\stg$ is indeed a lower bound (Lemma~\ref{lem:alpha-lb}) with Lemmas~\ref{lem:small-stage-lb} and~\ref{lem:large-stg-lb}.

\begin{lemma} \label{lem:beta-ub}
    Any probability vector $\vec p \in \unset$ satisfies the upper bound
    \[\Pr[\pbrv(\vec p)\in N] \leq \beta_\stg.\]
\end{lemma}
\begin{proof}
    Since the lower bound in the case where $\ceil {w/2} \geq 3 \sqrt \stg$ is constant, a trivial upper bound of 1 on the probability is sufficient.
    When $\ceil{w /2} < 3 \sqrt \nu$, we need a tighter bound.

    Fix any $\vec p \in \unset$.
    Let $Z \sim \normaldist(\mu, \sigma^2)$.
    By properties of the normal distribution, we have $\Pr[Z\in N] \leq \frac{w}{\sigma\sqrt{2\pi}}$.
    From Proposition~\ref{thm:normal-apx-to-pbd}, we can approximate the cumulative distribution function of a \gls{pbd} with an additive loss of $0.7915/\sigma$.
    Therefore, approximating $\Pr[\pbrv(\vec p) \in N]$ with $\Pr[Z\in N]$ incurs an additive loss of $\frac{2\cdot 0.7915}{\sigma}$.
    Therefore,
    \begin{equation*}
        \Pr[\pbrv(\vec p) \in N] \leq \frac{w}{\sigma\sqrt{2\pi}} + \frac{1.583}{\sigma} \leq \frac{w\sqrt2}{\sqrt{2\pi} \sqrt{\epsilon \nu }} + \frac{1.583\sqrt2}{\sqrt{\epsilon \nu }} \leq \frac{2w\sqrt2}{\sqrt{\epsilon\stg}}.
    \end{equation*}
    where the second inequality uses Lemma~\ref{lem:var-bound} and the last inequality holds because $w \geq 1$.
\end{proof}

\begin{corollary} \label{cor:overlap-prob-bound}
    For all stages $\stg \in [n]$ and $\vec p$ such that $\sum_{i=1}^\stg p_i \in \tn$,
    \begin{equation*}
        \Pr[\pbrv(\vec p) \in N] \ge \Omega(\sqrt\epsilon)\Pr[\pbrv(\vec p^*) \in N]
    \end{equation*}
    where $\vec{p}^*\in \unset$ is the optimal setting that maximizes the probability of not stopping \eqref{eq:stagewise-noncompletion}.
\end{corollary}
\subsection{Adversary Problem} \label{sec:adv-apx}
We now provide a $\bigOh{1/\sqrt{\epsilon}}$-approximation algorithm for \gls{adv} (Definition~\ref{def:adv-defn}).
Recall that for \gls{adv}, the sequence of testing $\sol$ is given, and the goal is to set $\vec p\in \unset$ so that the expected cost of testing, $C(\sol, \vec p)$, is maximized.
The algorithm uses results from Section~\ref{sec:single-stg-apx}, setting the vector $\vec p$ based on how the windows overlap.

Assume by renumbering that $\sol = \langle 1, 2, \dots, n\rangle$.
The cost of testing (see~\eqref{eq:cost-fn}) can be rewritten as 
\begin{equation} \label{eq:cost-restate}
    C(\sol, \vec p)  =\sum_{\stg=1}^n c_\stg \Pr[\pbrv_{\stg-1} (\vec p) \in N_{\stg-1}].
\end{equation}
The results of the previous section tell us how to set $(p_1, p_2, \dots, p_\stg)$ to maximize the probability term in \eqref{eq:cost-restate} for each stage $\stg$.
In this section, we show how to use the stage-wise result to maximize $C(\sol, \vec p)$.
This requires us to ensure consistency in the choice of $\vec p$ over stages.

The approximation algorithm considers the relative position of the windows $E_{n} = [\sum_{i=1}^{n}\ell_i, \sum_{i=1}^{n} r_i]$ and $ \tn_{n} \coloneq [N_{n}^\ell -1, N_{n}^h +1] = [k-1, k]$ in the last stage.
When $E_{n}^\ell > \tilde N_{n}^h$, the window $E_{n}$ is above the window $\tilde N_n$, thus, to increase the number of tests conducted, the adversary would want to set $\vec p = \vec \ell$ so that the expected number of passes $\sum_{i=1}^\stg p_i$ is close to, or within the non-stopping window for every stage $\stg$.
Symmetrically, if $E_{n}$ is below $\tn_{n}$, the adversary would want to set $\vec p = \vec r$.
In all other cases, $E_{n}$ and $\tilde N_{n}$ overlap ($E_{n} \cap \tilde N_{n} \neq \emptyset$).
As we will show, there exists a setting $\hat{\vec p}$ such that the expected number of passes lies in $\tilde N_\stg$ for all stages $\stg$ i.e., $\Pi_\stg(\hat{\vec p}) \coloneq \sum_{i=1}^\stg \hat p_i \in \tilde N_\stg$ for all $\stg \in [n]$.
We will show that any such setting of $\hat{\vec p}$ is an $\bigOh{1/\sqrt \epsilon}$ approximation.
A pseudocode is provided in Algorithm~\ref{alg:apx-adv}.

\begin{algorithm}
\caption{General Cost Adversary Algorithm} \label{alg:apx-adv}
\begin{algorithmic}[1]
    \State Compute $E^\ell_{n} = \sum_{i=1}^{n} \ell_i$ and $E^h_{n} =\sum_{i=1}^{n} r_i$.
    \State Compute the approximate solution $\hat{\vec p}$ based on the relative position of $E_{n}$ and $\tn_n$
    \If{$E^\ell_{n} >  \tn^h_{n} (=k)$} \Comment{$E_{n}$ lies above $\tn_{n}$}
        \State $\hat{\vec p} = (\ell_1, \ell_2, \dots, \ell_n)$.
    \ElsIf{$E^h_{n} < \tn^\ell_{n}(=k-1)$} \Comment{$E_{n}$ lies below $\tn_{n}$}
        \State $\hat{\vec p} = (r_1, r_2, \dots, r_n)$.
    \Else \Comment{$E_{n}$ and $\tn_{n}$ overlaps.}
        \State \label{line:window-straddle} Set $\hat{\vec p}$ such that $\sum_{i=1}^\stg \hat p_i \in \tn_\stg$ for all stages $\stg \in [n]$.
    \EndIf
\end{algorithmic}
\end{algorithm}
\begin{remark}
    While our algorithm considers the last-stage windows $E_n$ and $\tn_n$, the same result can be derived by considering the second-last-stage windows $E_{n-1}$ and $\tn_{n-1}$.
    In fact, the expression in \eqref{eq:cost-restate} does not even consider the non-stopping probability at stage $n$.
Using the stage $n$ window is crucial for our \gls{drst} algorithm because $E_n$ is the same regardless of the permutation $\sol$.
\end{remark}

We show that Algorithm~\ref{alg:apx-adv} is a $\bigOh{1/\sqrt \epsilon}$-approximation for \gls{adv}.
As with Section~\ref{sec:single-stg-apx}, the proof proceeds by cases. Here, the cases are whether $E_{n}$ is below, above or overlaps with $\tn_{n}$. 
\paragraph{Disjoint Windows}
When $E_{n}$ is above $\tn_{n}$, we first show that $E_{\stg}$ is never below $\tilde N_{\stg}$ for all stages $\stg \in [n]$.
We then show that setting $\vec p = \vec \ell$ is a $\bigOh{1/\sqrt \epsilon}$-maximizes the stagewise non-completion probability for all stages $\stg\in [n]$. 

\begin{lemma} \label{lem:below-once}
    If there is a stage $\stg\in [n]$ with $E_\stg^h < \tn_\stg^\ell$, then $E_{n}^h <\tn_{n}^\ell$.
\end{lemma}
\begin{proof}
    We can only have $E_\stg^h < \tn_\stg^\ell$ when $\tn_\stg^\ell \ge 1$, so we are in the case where $\tn_\stg^\ell = \stg - n + k - 1$.
    We have
    \begin{align*}
        E_{n}^h &=E_\stg^h + \sum_{i=\stg+1}^{n} r_i
        < \tn_\stg^\ell + \sum_{i=\stg+1}^{n} r_i 
        \leq (\stg  - n + k - 1)  + \sum_{i=\stg +1}^{n} r_i\\
        &\le \stg  - n + k - 1  + (1-\epsilon)(n-\stg)
        = k - 1 - \epsilon (n-\stg)
        \leq k-1
        = \tn_{n}^\ell,
    \end{align*}
    where the first inequality is by our assumption that $E_\stg^h < \tn_\stg^\ell$, the second uses the definition of $\tn_\stg^\ell$, and the third uses $r_i \in [\epsilon, 1-\epsilon]$.
\end{proof}

\begin{lemma} \label{lem:adv-disjoint}
    If $E_{n}^\ell > \tn_{n}^h$, then setting $\vec p = \vec \ell$ is a $\bigOh{1/\sqrt \epsilon}$ approximation for \gls{adv}.
\end{lemma}
\begin{proof}
    Lemma~\ref{lem:below-once} implies that there are no stages $\stg \in [n]$ where $E_\stg$ is below $\tn_\stg$; at every stage $\stg$, either $E_\stg$ is above, or overlaps $\tn_\stg$.
    The first consequence of this is that the set \(O = \{\stg \in  [n] : E_{\stg-1}^\ell \leq \tn_{\stg-1}^h \}\) is the set of stages where $E_{\stg-1}$ overlaps $\tn_{\stg-1}$.
    The second consequence is that in all other stages $\stg \in [n] \setminus O$,  $E_{\stg-1}$ is above $\tn_{\stg-1}$.

    Since the algorithm sets $\vec p = \vec \ell$, the number of expected passes observed is $\sum_{i=1}^\stg \ell_\stg = E_\stg^\ell$.
    This implies that if $\stg \in O$, then $\sum_{i=1}^{\stg-1} \ell_i \in N_{\stg-1}$.
    Therefore,
    \begin{align*}
        C(\sol, \vec \ell) &= \sum_{\stg=1}^{n} c_{\stg} \Pr[\pbrv_{\stg-1}(\vec \ell) \in N_{\stg-1}] \\
        &= \sum_{\stg \in O} c_{\stg} \Pr[\pbrv_{\stg-1}(\vec \ell) \in N_{\stg-1}] + \sum_{\stg \in [n] \setminus O} c_{\stg} \Pr[\pbrv_{\stg-1}(\vec \ell) \in N_{\stg-1}] \\
        &\geq \sum_{\stg \in O} c_{\stg}\cdot  \Omega(\sqrt\epsilon) \max_{\vec p \in \unset} \Pr[\pbrv_{\stg-1}(\vec p) \in N_{\stg-1}] + \sum_{\stg \in [n] \setminus O} c_{\stg} \cdot \max_{\vec p \in \unset} \Pr[\pbrv_{\stg-1}(\vec p) \in N_{\stg-1}]\\
        &\geq \sum_{\stg \in O} c_{\stg} \cdot \Omega(\sqrt \epsilon)\Pr[\pbrv_{\stg-1}(\vec p^*) \in N_{\stg-1}] + \sum_{\stg \in [n] \setminus O} c_{\stg} \Pr[\pbrv_{\stg-1}(\vec p^*) \in N_{\stg-1}] \\
        &\geq \Omega(\sqrt\epsilon)\cdot \sum_{\stg=1}^{n} c_\stg \Pr[\pbrv_{\stg}(\vec p^*) \in N_{\stg-1}]\\
        &= \Omega(\sqrt\epsilon) \cdot C(\sol, \vec p^*).
    \end{align*}
    where \(\vec p^* = \arg\max_{p\in \unset} c(\sol, \vec p)\) denote the optimal solution to \gls{adv} for the sequence $\sol$. Above, the first inequality uses Corollary~\ref{cor:overlap-prob-bound} for the case where $\stg \in O$ and Lemma~\ref{lem:disjoint-optimal-above} for the case where $\stg \notin O$ (the first and second term respectively). 
    The second inequality uses the fact that $\vec p^*$ is a feasible solution to the maximization problem for all stages $\stg \in [n-1]$.
\end{proof}

The proof when $E_{n}$ is below $\tn_{n}$ is symmetrical. 
\begin{lemma} \label{lem:adv-below}
    If $E_{n}^h < \tn_{n}^\ell$, then setting \(\vec p = \vec \ell\) is a $\bigOh{1/\sqrt \epsilon}$ approximation for \gls{adv}.
\end{lemma}
\begin{proof}
    Instead of using Lemma~\ref{lem:below-once}, we use the fact that since $E^\ell_\stg$ is monotone increasing in $\stg$. If there is a stage $\stg\in [n]$ where $E_\stg^\ell > k-1$, then $E_{n}^\ell > k-1$.
    By our assumption that $E_n^h \leq \tn_n^\ell (=k-1)$, there are no stages $\stg \in [n]$ where $E_\stg$ is above $\tn_\stg$.
    We can hence partition the stages into $O = \{\stg \in [n]: E_{\stg -1}^h \geq \tn_{\stg-1}^\ell\}$, where the windows $E_{\stg-1}$ and $\tn_{\stg-1}$ overlaps, and remaining stages $[n] \setminus O$, where $E_{\stg-1}$ is below $\tn_{\stg-1}$.
    Repeating a similar analysis as Lemma~\ref{lem:adv-below} yields $C(\sol, \vec r) \geq \Omega(\sqrt\epsilon)\cdot C(\sol, \vec p^*)$ as desired.
\end{proof}

\paragraph{Overlapping Windows} 
We show that in the overlapping case, we can find a $ {\vec p} \in \unset$ such that the induced path $\Pi({\vec p})$ is always in $\tilde N$ at every stage $\stg$.
This vector $ {\vec p}$ is a $O(1/\sqrt\epsilon)$ approximation for \gls{adv}.
In the next proof, we use the following equivalence: for any stage $\stg \in [n]$,
\[E_\stg \cap \tn_\stg \neq \emptyset \iff (E_\stg^\ell\leq \tn_\stg^h) \land (E_\stg^h \geq \tn_\stg^\ell).\]

\begin{lemma} \label{lem:overlap-throughout}
    If $E_{n} \cap \tn_{n} \neq \emptyset$, then for all $\stg \in [n-1]$, $E_\stg \cap \tn_\stg \neq \emptyset$.
\end{lemma}
\begin{proof}

    Since $E_n^\ell \leq \tn_n^h(=k)$ and $E_\stg^\ell$ is non-decreasing in $\ell$, we know for all stages $\stg \in [n]$, $E_{\stg}^\ell \leq k$ as well.
    Moreover, by our assumption that $p_i \geq \epsilon$ and $n \geq \frac 2\epsilon$, we have $E_{\floor{n/2}}^\ell \leq k-1$, so $E_\stg^\ell \leq k-1$ for all $\stg \in [\floor{\frac n2}]$.
    Thus, $E_\stg^\ell \leq \tn_\stg^h$ for all stages $\stg \in [n-1]$.

    Moreover, we know $E_n^h \geq N_n^\ell$.
    By Lemma~\ref{lem:below-once}, $E_\stg^h \geq \tn_\stg^\ell$ for all $\stg \in [n]$.
    Since $E_\stg^\ell \leq \tn_\stg^h$ and $E_\stg^h \geq \tn_\stg^\ell$ for all $\stg \in [n]$, we have $E_\stg\cap \tn_\stg \ne \emptyset$ for all $\stg \in [n]$.
    
\end{proof}

\begin{lemma} \label{lem:path-in-window-exists}
    Suppose \(E_{n} \cap \tn_{n} \neq \emptyset\).
    there is a setting \(\vec p \in \unset \) such that \(\Pi_\stg(\vec p) \in \tn_\stg\) for all \(\stg \in [n]\).
    Moreover, we can find such $\vec p$ in polynomial time.
\end{lemma}
\begin{proof}
    We prove the lemma by proving a stronger statement:
    for every stage $\stg\in [n-1]$, for every target $T \in E_\stg\cap \tn_\stg$, there is a probability vector $\vec{p} \in \reals^\stg$, with $p_i \in [\ell_i, r_i]$ for every $i\in [\stg]$ such that $\sum_{i=1}^\stg p_i = T$.
    By Lemma~\ref{lem:overlap-throughout}, $E_{n} \cap \tn_{n} \neq \emptyset$, implies that $E_{\stg} \cap \tn_{\stg} \neq \emptyset$ for all $\stg \in [n]$, so $T \in E_\stg \cap \tn_\stg$ exists in every stage.

    We prove the statement by induction on the stage.
    In the base case, for every target $T \in E_1\cap \tn_1$, the setting $p_1 = T$ is feasible since $T\in E_1$.
    For the induction step, we assume that the proposition holds for some stage $\stg - 1$ for  $\stg \in \{2, 3, \dots, n\}$.
    Consider some target $T \in E_\stg \cap \tn_\stg$.
    It suffices to show the claim:
    \begin{equation} \label{eq:ind-claim}
        \text{There exists some $p_\stg\in [\ell_\stg, r_\stg]$ such that $T' = T - p_\stg \in E_{\stg-1}$.}
    \end{equation}
    Then, by our inductive hypothesis, we can reach $T'$ using some $\vec{p} \in \reals^{\stg-1}$ with $p_i \in [\ell_i, r_i]$ for every $i\in [\stg-1]$ such that
  $\sum_{i=1}^{\stg-1}{p}_i = T'$.

    To prove \eqref{eq:ind-claim}, it suffices to show that the following set of linear inequalities has a solution $p_\stg$:
    \begin{gather*}
        \ell_\stg \leq p_\stg \leq r_\stg\\
        E_{\stg-1}^\ell \leq T-p_\stg \leq E_{\stg-1}^h \label{eq:prev-in-window}
    \end{gather*}
    This is equivalent to verifying that 
    \begin{equation*}
        \max\{T - E_{\stg-1}^h, \ell_\stg\} \leq \min\{T-E_{\stg-1}^\ell, r_\stg\}.
    \end{equation*}

    Indeed, we have $T-E_{\stg-1}^h = T - (E_\stg^h-r_i) = (T-E_\stg^h)+r_i \leq r_i$, since $(T-E_\stg^h)\leq 0$ and $T-E_{\stg-1}^\ell = T - (E_\stg^\ell-l_i) = (T-E_\stg^h)+\ell_i \geq \ell_i$, since $(T-E_\stg^\ell)\geq 0$.

Moreover, a feasible value for $p_\stg$ in \eqref{eq:ind-claim} can easily be obtained in polynomial time, using the above inequality (where the values of $T, E_{\stg-1}^h, E_{\stg-1}^\ell, \ell_\stg$ and $r_\stg$ are fixed).  
\end{proof}

\begin{lemma}\label{lem:adv-overlap}
    Suppose \(E_{n} \cap \tn_{n} \neq \emptyset\), there is a \(O(1/\sqrt\epsilon)\)-approximation for \gls{adv}.
\end{lemma}
\begin{proof}
    Define $\vec {\hat p}$ such that $\vec {\hat p} \in \unset$ and $\Pi_\stg(\vec {\hat p}) \in \tn_\stg$ for all $\stg \in [n]$.
    We know that such a setting is feasible by Lemma~\ref{lem:path-in-window-exists}. Let
    $\vec p^* = \arg\max_{\vec p \in \unset} c(\sol, \vec p)$,
    \begin{align*}
        C(\sol, \vec {\hat p}) &=
        \sum_{\stg=1}^n c_\stg \Pr[\pbrv_{\stg-1}(\vec {\hat p}) \in N_{\stg-1}] \\
        &\ge \sum_{\stg=1}^n c_\stg\cdot\Omega(\sqrt \epsilon)  \max_{\vec p \in \unset} \Pr[\pbrv_{\stg-1}(\vec {p}) \in N_{\stg-1}] \\
        &\ge \Omega(\sqrt \epsilon)\cdot \sum_{\stg=1}^n c_\stg \Pr[\pbrv_{\stg-1}(\vec p^*) \in N_{\stg-1}] \\
        &\ge \Omega(\sqrt\epsilon) \cdot C(\sol, \vec p^*),
    \end{align*}
    where the first inequality uses Corollary~\ref{cor:overlap-prob-bound} and the second follows since $\vec p^*$ is a feasible solution to the maximization problem $\max_{\vec p \in \unset} \Pr[\pbrv_{\stg-1}(\vec p) \in N_{\stg-1}]$ for every $\stg\in [n]$.
\end{proof}
\subsection{Solving the distributionally robust problem} \label{sec:gen-rob-prob}
Having approximated \gls{adv}, we now  obtain an approximation algorithm for the distributionally robust \kofn sequential testing problem with general cost.

The algorithm is based on the approximate \gls{adv} algorithm from Section~\ref{sec:adv-apx}.
For a given test sequence $\sol$, the approximate adversary of Section~\ref{sec:adv-apx} looks only at the relative position of the last windows $E_{n}(\sol)$ and $\tn_{n}$ to determine the setting of $\vec p$.
However, the crucial property for \gls{drst} is that the relative position of $E_n(\sol)$ and $\tn_n$ is independent of the permutation $\sol$:

\begin{lemma} \label{lem:same-end-window}
    The last stage window $E_n(\sol)$ is the same regardless of the permutation $\sol$
\end{lemma}
\begin{proof}
    For the last stage $n$, we have  $E_{n}^\ell(\sol) = \sum_{i=1}^{n} \ell_{\sol_i}=\sum_{i=1}^n \ell_i$ as all tests are included in the summation (for each $\sol$). Similarly,  $E_{n}^h(\sol) = \sum_{i=1}^{n} r_{i}$. 
\end{proof}

Thus, we can write $E_n (\sol)$ as simply $E_n$.
When $E_{n}$ lies above $\tn_{n}$, we have shown in Lemma~\ref{lem:adv-disjoint} that an approximate adversary sets $\vec p = \vec \ell$ regardless of the permutation $\sol$. Thus, we can use any \(\rho\mkern1mu\)-approximation to the classical sequential testing problem with $\vec p = \vec \ell$ to obtain a permutation for \gls{drst}.
Symmetrically, if $E_{n}$ lies below $\tn_{n}$, then the (approximate) adversary sets $\vec p = \vec r$, and  again we obtain a permutation for \gls{drst} by solving an instance of classical sequential testing.
Finally, when $E_{n}$ and $\tn_{n}$ overlap, we show that sorting by decreasing cost is an approximate solution.
Let $\sst$ denote any $\rho\mkern1mu$-approximate algorithm to the (classical) $k$-of-$n$  testing problem, with known probabilities. The algorithm $\sst$ takes as input the cost-vector $\vec c$ and probability-vector   $\vec p$, and returns a $\rho\mkern1mu$-approximate solution $\asol$.
The pseudocode is presented in Algorithm~\ref{alg:drst}.

\begin{figure}
     \centering
     \begin{subfigure}[b]{0.32\textwidth}
         \centering
        \includegraphics[width=1\linewidth,page=1]{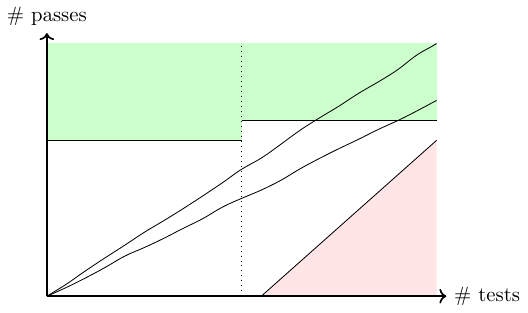}
         \caption{$E_n^\ell > \tn_n^h$}
         \label{fig:case-above}
     \end{subfigure}
     \hfill
     \begin{subfigure}[b]{0.32\textwidth}
         \centering
        \includegraphics[width=1\linewidth,page=2]{figures/drst-cases.pdf}
         \caption{$E_n^h < \tn_n^\ell$}
         \label{fig:case-below}
     \end{subfigure}
     \hfill
     \begin{subfigure}[b]{0.32\textwidth}
         \centering
        \includegraphics[width=1\linewidth,page=3]{figures/drst-cases.pdf}
         \caption{$E_n\cap \tn_n \neq \emptyset$}
         \label{fig:}
     \end{subfigure}
        \caption{Cases on $E_n$ and $\tn_n$}
        \label{fig:window-cases}
\end{figure}

\begin{algorithm}
\caption{General Cost Distributionally Robust Algorithm} \label{alg:drst}
\begin{algorithmic}[1]
\State Compute $E^\ell_{n}$ and $E^h_{n}$.
\State Compute the approximate solution $\asol$ based on the relative position of $E_{n}$ and $\tn_{n}$ 
\If{$E^\ell_n >  \tn^h_n (=k)$} \Comment{$E$ lies above $\tn$}
    \State $\asol = \sst(\vec c, \vec \ell)$.
\ElsIf{$E^h_n < \tn^\ell_n (=k-1)$} \Comment{$E$ lies below $\tn$}
    \State $\asol = \sst(\vec c, \vec r)$.
\Else \Comment{$E$ and $\tn$ overlaps.}
    \State $\asol$ is sorted by increasing order of testing cost $c$.
    i.e., $c_{\asol_1} \leq c_{\asol_2} \leq \dots \leq c_{\asol_n}$.
\EndIf
\end{algorithmic}
\end{algorithm}

Let $\asol$ be the solution of $\alg$.
Recall that $\asol$ is a permutation of $[n]$ specifying the order in which tests should be conducted.
We continue our analysis by cases depending on how the last-stage expected value window  $E_n$ overlaps with the non-stopping  window $\tn_n = [k-1, k]$.

\paragraph{Windows disjoint.} We first consider the case that $E_n\cap \tn_n=\emptyset$. 
Recall that $\adv(\asol) = \max_{\vec{p}\in \unset} C(\asol, \vec{p})$ is the optimal value achieved by \gls{adv} given a solution $\asol$.

\begin{lemma} \label{lem:drst-disjoint}
    Assuming a $\rho$-approximation for the classical \kofn sequential testing problem, there is a $O(\rho/\sqrt\epsilon)$ approximation for \gls{drst} when $E_n^\ell > \tn_n^h$.
\end{lemma}
\begin{proof}
Recall that our solution $\asol$ is a $\rho\mkern1mu$-approximate sequence of the \kofn sequential testing problem $\sst(\vec c, \vec \ell)$.
    Let $\osol$ denote the optimal solution for \gls{drst}. Then,
    \begin{align*}
        \adv(\asol) &= \max_{\vec p \in \unset} C(\asol, \vec p)  \,\,\leq\,\, \bigOh{\frac1{\sqrt \epsilon}} \cdot C(\asol, \vec \ell) \leq \bigOh{\frac1{\sqrt \epsilon}} \cdot  \rho\cdot C(\sigma,\vec \ell)  \,\leq\, \bigOh{\frac\rho{\sqrt \epsilon}}  \cdot  \adv(\osol)\notag. 
    \end{align*}
    The first inequality uses the adversary's approximation proved in Lemma~\ref{lem:adv-disjoint}. The second inequality follows since $\asol$ is a $\rho\mkern1mu$-approximation for \kofn testing with probabilities $\vec \ell$ and $\sigma$ is any other solution. 

\end{proof}

A symmetrical proof would give us
\begin{lemma} \label{lem:drst-disjoint-below}
    Assuming a $\rho$-approximation for the classical \kofn sequential testing problem, there is a $O(\rho/\sqrt\epsilon)$ approximation for \gls{drst} when $E_n^h < \tn_n^\ell$.
\end{lemma}

\paragraph{Windows overlap.} We now consider the case that $E_n\cap \tn_n\ne \emptyset$. 
Here, we will show that it suffices to perform tests in increasing order of costs. 
The proof uses the lower and upper bounds ($\alpha_\stg$ and $\beta_\stg$ from Definition~\ref{def:alpha-beta-bounds}) of the non-completion probabilities to bound the expected cost.
One crucial property is that $\beta_\stg$  is non-increasing. 
\begin{lemma} \label{lem:beta-noninc}
    The function $\beta_\stg$ is non-increasing.
\end{lemma}
\begin{proof}
Recall that $\beta_\stg=\min\{ 2\sqrt 2\cdot g(\stg) ,1\}$ where $g(\stg)=\frac{\min\{k-1,\, n - \stg \}}{\sqrt{\epsilon \stg}}$. We now show  that $g(\stg)$ is non-increasing, which suffices for the lemma. Indeed, the numerator of $g$ is a non-increasing function $\min\{k-1,\, n - \stg\}$ and the denominator is an increasing function $\sqrt{\epsilon \stg}$. 
\end{proof}

\begin{lemma} \label{lem:drst-overlap-alg}
    The permutation \( \asol \) that sorts by increasing cost, i.e., \( c_{\asol_1} \leq c_{\asol_2} \leq \dots \leq c_{\asol_n}\) is a \( \bigOh{1/\sqrt \epsilon} \) approximation for the case where \( E_n\) and \( \tn_n \) overlap.
\end{lemma}
\begin{proof}
    Let $\asol$ denote the increasing cost permutation. Without loss of generality (by renumbering tests), let $\asol=\langle 1, 2,\dots, n\rangle$ so $c_1 \leq c_2 ,\leq \dots \leq c_n$. Let $\vec p^* =\arg\max_{\vec p \in \unset} C(\asol, \vec p)$ be the solution to $\adv$ given the permutation $\asol$. Let $\sigma$ denote the optimal \gls{drst} permutation.  
    We can upper-bound the adversary's cost on the increasing-cost permutation $\asol$
    \[ \adv(\asol) = C(\asol, \vec p^*) = \sum_{\stg =1}^n c_\stg \Pr[\pbrv_{\stg-1}(\vec p^*) \in N_{\stg-1}] \leq \sum_{\stg=1}^n c_\stg \beta_{\stg-1} \le \sum_{\stg=1}^n c_{\osol_\stg} \beta_{\stg-1}, \]
     The first inequality uses Lemma~\ref{lem:beta-ub} (note that $\beta_\stg$ depends only on $\stg$ and not $\sol$). The second inequality uses the fact that $c_\stg$ is non-decreasing and $\beta_{\stg-1}$ is non-increasing (Lemma~\ref{lem:beta-noninc}).
     To lower bound $\opt$, we define $\hat{\vec p}$ such that $\Pi(\hat{\vec p})$ is always in the non-stopping window.
    Such a choice of $\hat{\vec p}$ is always feasible by Lemma~\ref{lem:path-in-window-exists}. Hence,
    \[ \adv(\osol) \geq C(\osol,{\hat{\vec p}})\geq \sum_{\stg=1}^n c_{\osol_{\stg}}\alpha_{{\stg-1}}, \]
    where the second inequality uses Lemma~\ref{lem:alpha-lb}.
    Using the definition of $\alpha_\stg$ and $\beta_\stg$ (Definition~\ref{def:alpha-beta-bounds}), we have $\frac{\beta_\stg}{\alpha_\stg}=O(\frac{1}{\sqrt{\epsilon}})$, which implies:
    \[\adv(\asol) \leq \bigOh{\frac 1{\sqrt \epsilon}} \cdot\adv(\osol)\]
\end{proof}

We can now complete the proof of Theorem~\ref{thm:gen-cost}. 
    If the windows are disjoint, then by Lemmas~\ref{lem:drst-disjoint} and~\ref{lem:drst-disjoint-below}, we have a $\bigOh{\frac \rho{\sqrt\epsilon}}$-approximation.
    Moreover, by \cite{gkenosisStochasticScoreClassification2018a}, we know that $\rho \leq 2$, so we have a $\bigOh{\frac 1{\sqrt\epsilon}}$ approximation in the disjoint case.
    On the other hand, if the window overlaps, then by Lemma~\ref{lem:drst-overlap-alg}, we have an $O\left(\frac1{\sqrt\epsilon}\right)$ approximation.

\section{Better Approximation for the Adversary Problem}
In this section, we show a better approximation for \gls{adv} (Definition~\ref{def:adv-defn}).
Recall that for \gls{adv}, we fix some sequence $\sol$ to conduct the tests.
The goal is to choose $\vec p \in \unset$ such that the expected cost of testing $C(\sol, \vec p)$ is maximized. It is easy to see that the optimal probability vector $\vec{p^*}$ sets each $p^*_i\in \{\ell_i , r\}$. So, a naive enumeration leads to a $2^n$ time exact algorithm. In this section, we obtain a \gls{qptas} for \gls{adv}. Throughout this section, we assume that the given sequence is $\sol=\langle 1,2,\dots, n\rangle$. 

Our algorithm uses a strong structural result of \gls{pbd}s from
\cite{daskalakis_sparse_2015},  which states that two \glspl{pbd} are close in terms of  the \gls{tv} distance  as long as their first few  moments are equal.

\begin{theorem}{\cite[Theorem~3]{daskalakis_sparse_2015}}\label{thm:moment-apx}
    Let $\vec p , \vec{q}\in [0, 1/2]^n$  be probability vectors such that their  first $d$ power-sums are equal, i.e., for all $a\in \{1,2,\dots, d\}$ we have 
    \begin{equation} \label{eq:equal-power-sum}
        \sum_{i=1}^n p_i^a = \sum_{i=1}^n q_i^a.
    \end{equation}
Then,  
$\tvdist \left(\pbrv(\vec p), \pbrv(\vec q )\right) \leq 13(d+1)^{1/4} 2^{-(d+1)/2}$. 
\end{theorem}

As stated in Remark~1 of \cite{daskalakis_sparse_2015}, condition \eqref{eq:equal-power-sum} of Theorem~\ref{thm:moment-apx} is equivalent to equating the first $d$ moments.
Moreover, Theorem~\ref{thm:moment-apx} also holds when both $\vec p , \vec q\in [1/2, 1]^n$.

\paragraph{Rounding the initial instance.} We first show that all end-points $\{\ell_i ,r_i\}_{i=1}^n$ can be rounded to be integer multiples of $\frac{1}{n^3}$, at a small loss in the optimal value.  The key observation is the following:
\begin{lemma} \label{lem:rounding-loss}
Let $\vec p, \vec{p'} \in [0,1]^n$ be two  probability vectors such that $\abs{p_i-p'_i}<\frac1{n^3}$ for all $i\in [n]$. Then, 
$\tvdist\left(\pbrv(\vec p), \pbrv(\vec p')\right) \leq \frac{1}{n^2}$ and hence $\abs{C(\sol,\vec{p}) - C(\sol,\vec{p'}) }\le \frac{c_{\max}}n$. 
\end{lemma}
\begin{proof}
    We use the bound on the total variation distance between product distributions,
    \begin{align*}
        \tvdist\left(\pbrv(\vec p), \pbrv(\vec p')\right) \leq \sum_{i=1}^n \tvdist(\bern(p_i), \bern(p'_i)) = \sum_{i=1}^n \abs{p_i-p'_i} \le    \frac{1}{n^2}.
    \end{align*}
    We can now bound the difference in expected cost 
    $$\abs{C(\sol,\vec{p}) - C(\sol,\vec{p'})  } \le n\cdot c_{\max} \cdot 
        \tvdist\left(\pbrv(\vec p), \pbrv(\vec p')\right)  \le \frac{c_{\max}}{n},$$
        where the first inequality uses that the cost is always bounded in  $[0 , n  c_{\max} ]$. 
\end{proof}
Given any instance, we round down (resp.\@ up)  each $\ell_i$ (resp.\@ $r_i$) to the grid $\frac{1}{n^3}\cdot \Z$. Let $\opt$ and $\vec{p}$ (resp.\@ $\opt'$ and $\vec{q'}$) denote the optimal value and solution of the original (resp.\@ rounded) instance. Let $\vec{p'}$ denote the rounded version of $\vec{p}$ and let $\vec{q}$ denote the original probability vector corresponding to $\vec{q'}$. 
Then,
$$\opt-\opt' = C(\sol,\vec{p}) - C(\sol,\vec{q'}) =  C(\sol,\vec{p}) - C(\sol,\vec{q})  + C(\sol,\vec{q}) - C(\sol,\vec{q'}) \ge C(\sol,\vec{q}) - C(\sol,\vec{q'})  \ge -\frac{c_{\max}}{n}.   $$
The first inequality uses the fact that $\vec{p}$ is the optimal solution to the original instance and $\vec{q}$ is a feasible solution. The last inequality uses Lemma~\ref{lem:rounding-loss}. Similarly, we have $\opt-\opt'\le \frac{c_{\max}}{n}$. So, $\abs{\opt-\opt'}\le \frac{c_{\max}}{n}$.

Henceforth, we  work with rounded instances of \gls{adv}. 

\paragraph{Naive dynamic program.} We first present the exact enumerative  algorithm as a  dynamic program. (This view is helpful in introducing   the more efficient approximation algorithm.) The DP proceeds in $n$ stages, where for each  $i\in [n]$, the state space is 
$$\Omega_i =\{ (\omega_1,\dots, \omega_{i-1}) \,:\, \omega_j\in \{\ell_j , r_j\} \mbox{ for all }1\le j\le i-1 \}. $$
The cost at any state $\vomega\in \Omega_i$ is defined as 
\begin{equation} \label{eq:qptas-stg-cost}
    c(\vomega)=c_i\cdot \Pr[\pbrv(\omega_1,\dots, \omega_{i-1}) \in N_{i-1}],
\end{equation}
which accounts for the cost incurred by test $i$ assuming the probabilities of  prior tests are as in $\vomega$. 

The value function in any stage $i\in [n]$ is given recursively as:
$$V_i(\vomega) \quad =\quad \max_{p_i\in \{\ell_i , r_i\}} \quad \left\{ c(\vomega) + V_{i+1}(\vomega\circ p_i)  \right\},$$
where we use the boundary condition $V_{n+1}(\star)=0$. Clearly, the optimal value is given by $V_1(\emptyset)$.

\paragraph{Compressed dynamic program.} Note that the state space of the exact DP is $2^n$, which is exponential. We will use Theorem~\ref{thm:moment-apx} to compress the state space to quasi-polynomial size, at a small loss in the objective.
In particular, at each stage $i$ we will just track the first $d=O(\log n)$ power-sums of the current probability vector $\vomega\in \Omega_i$. One technical issue is that Theorem~\ref{thm:moment-apx} only applies to probability vectors in either $[0,\frac12]^n$ or  $[\frac12,1]^n$. In order to handle this, we will maintain two vectors  of power-sums: one for low probability values (below $\frac12$) and another for high  probability values (above $\frac12$). Formally, we map each $\vomega\in \cup_{i=1}^n \Omega_i$ to the pair $\pi(\vomega) = (\vec m_L, \vec m_H)\in \mathbb{R}^{2d}$ where
$$m_L(a) = \sum_{j<i:\omega_i\le \frac12} \omega_i^a \quad \mbox{and} \quad m_H(a) = \sum_{j<i:\omega_i>\frac12} \omega_i^a ,\quad \forall a\in [d]. $$
The compressed state-space  is the set of all such vectors arising when the probabilities lie in the grid  $\frac{1}{n^3}\cdot \Z$. Formally,
$$\Pi = \left\{ (\vec{u}, \vec{v})\in [0,1]^{d}\times [0,1]^{d} \,:\, u_a , v_a\in \frac{\Z}{n^{3a}} \mbox{ for all }a\in [d]\right\}.$$
Crucially, the size of this state space is only exponential in $d$ (and not $n$).

\begin{lemma}
    The cardinality $\abs{\Pi}= n^{O(d^2)}$.
\end{lemma}
\begin{proof}
Consider any $(\vec{u}, \vec{v})\in \Pi$. For any  $a\in [d]$, the values $u_a$ and $v_a$ are of the form $\sum_{i\in S} p_i^a$.
where $S\subseteq [n]$ and each probability $p_i\in \frac{1}{n^3}\cdot \Z$. It follows that both $u_a$ and $v_a$ are of the form $\frac{z}{n^{3a}}$ where $z\in \{0,1,\dots, n^{3a}\}$. Therefore, the number of possible vectors is 
$$\prod_{a=1}^d n^{3a}\times n^{3a} = n^{O(d^2)},$$
which completes the proof. 
\end{proof}

For each $i\in [n]$, we define the {\em feasible states} as the image of $\Omega_i$ under $\pi$, i.e.,  
$$\Pi_i = \left\{ \pi(\vomega) : \vomega\in \Omega_i \right\}\subseteq \Pi.$$
For each $(\vec{u}, \vec{v})\in \Pi_i$ we arbitrarily choose a {\em representative} probability vector $r(\vec{u}, \vec{v}) =\langle p_1,\dots, p_i\rangle$ that maps to  $(\vec{u}, \vec{v})$. 

The next lemma shows that all states in $\Omega_i$  mapping to the same vector $(\vec{u}, \vec{v})$ are essentially  equivalent.  
\begin{lemma}\label{lem:adv-qptas-1}
    For any 
    $\vomega\in \Omega_i$ 
    we have 
    $$ \tvdist \left(\pbrv(r(\pi(\vomega))), \pbrv(\vomega )\right) \,\,\leq\,\, \epsilon \,:=\, O(d\, 2^{-d/2}). $$
\end{lemma}
\begin{proof}
    Let  
    $\vec{p} = \langle p_1,\dots, p_i\rangle$ denote the probabilities in $r(\pi(\vomega))$ and $\langle \omega_1,\dots, \omega_i\rangle$ those in $\vomega$. Also, let $\vec{p}_L$ (resp.\@ $\vomega_L$) consist of only those probabilities in $\vec{p}$ (resp. $\vomega$) that are at most half. Similarly, let $\vec{p}_H$ (resp. $\vomega_H$) consist of only those probabilities in $\vec{p}$ (resp. $\vomega$) that are more than half. Let  $\pi(\vomega)=(\vec{u}, \vec{v})$. Observe that the first $d$ power-sums of $\vec{p}_L$ and $\vomega_L$ are equal (given by $\vec{u}$). Also, the first $d$ power-sums of $\vec{p}_H$ and $\vomega_H$ are equal (given by $\vec{v}$). Applying Theorem~\ref{thm:moment-apx} for both the low and high cases,
 $$ \tvdist \left(\pbrv(\vec{p}_L), \pbrv(\vomega_L )\right)\,\mbox{ and }\,  \tvdist \left(\pbrv(\vec{p}_H), \pbrv(\vomega_H )\right) \,\,\leq\,\, O(d\, 2^{-d/2}). $$
  We now use  $\pbrv(r(\pi(\vomega)))= \pbrv(\vec{p})=\pbrv(\vec{p}_L) + \pbrv(\vec{p}_H)$ and $\pbrv(\vomega )=\pbrv(\vomega_L )+\pbrv(\vomega_H )$, where  the ``low'' and ``high'' components are independent in both cases. Using the TV-distance upper bound for sums of independent r.v.s, 
 $$ \tvdist \left(\pbrv(r(\pi(\vomega)))), \pbrv(\vomega )\right) \le  \tvdist \left(\pbrv(\vec{p}_L), \pbrv(\vomega_L )\right)+  \tvdist \left(\pbrv(\vec{p}_H), \pbrv(\vomega_H )\right)= O(d\, 2^{-d/2}),$$
 which completes the proof. \end{proof}
We now define costs at  compressed states. For any $(\vec{u}, \vec{v})\in \Pi_i$, its cost 
$\bar{c}(\vec{u}, \vec{v}) = c(r(\vec{u}, \vec{v}))$ which corresponds to the cost of its representative state in the  exact DP. 

\paragraph{Efficient computation of $\bar{c}$.} Given  any state $(\vec{u}, \vec{v})$ in stage $i$, we can find a  representative vector $r(\vec{u}, \vec{v})\in  \{\ell_1 ,r_1\}\times \{\ell_2 ,r_2\} \times \dots \times\{\ell_i ,r_i\}$  (if one exists) by a polynomial-size DP (with only $i$ stages). Then, the cost  $\bar{c}(\vec{u}, \vec{v}) =  c(r(\vec{u}, \vec{v}))$ (see \eqref{eq:qptas-stg-cost})  can be computed exactly using another DP. 

\medskip

The new value function in stage $i$ is given by
$$D_i(\pi(\vomega))   \quad =\quad \max_{p_i\in \{\ell_i , r_i\}} \quad \left\{ \bar{c}(\pi(\vomega)) + D_{i+1}(\pi(\vomega\circ p_i))  \right\}.$$
If $\pi(\vomega)= (\vec{u}, \vec{v})$ then $\pi(\vomega\circ p_i)$ equals $(\vec{u'}, \vec{v'})$ where
$$
u'_a =\left\{
\begin{array}{ll}
   u_a+p_i^a  & \mbox{ if }p_i\le \frac12 \\
    u_a & \mbox{otherwise}
\end{array}\right. \quad \mbox{ and }\quad 
v'_a =\left\{
\begin{array}{ll}
   v_a+p_i^a  & \mbox{ if }p_i>\frac12 \\
    v_a & \mbox{otherwise}
\end{array}\right. ,\quad \forall a\in [d].
$$
So, the next state $(\vec{u'}, \vec{v'})$ can be calculated easily from the current state $(\vec{u}, \vec{v})$.

The optimal value of the DP is $D_1(\emptyset)$ in stage $1$ (we use the boundary condition $D_{n+1}(\star)=0$).

We  now relate the value functions in the two DPs.
\begin{lemma}\label{lem:adv-qptas-2}
    For each stage $i$ and $\vomega\in \Omega_i$, we have $\abs{D_i(\pi(\vomega))-V_i(\vomega)} \le (n+1-i)\cdot \epsilon \, c_{\max}$
\end{lemma}
\begin{proof}
    We proceed by (backward) induction on $i$. The base case is $i=n+1$ where  $D_{n+1}(\star)=V_{n+1}(\star)=0$. Suppose that the lemma is true for stage $i+1$ and consider any $\vomega\in \Omega_i$. We have:
    \begin{eqnarray}
        \abs{c(\vomega)-\bar{c}(\pi(\vomega))} & = &c_i \, \abs{ \Pr[\pbrv(\vomega)\in N_{i-1}] - \Pr[\pbrv(r(\pi(\vomega)))\in N_{i-1}] } \notag\\
        & \le & c_{\max} \cdot  \tvdist \left(\pbrv(r(\pi(\vomega))), \pbrv(\vomega )\right) \,\,\le \,\, \epsilon\, c_{\max}.        \label{eq:adv-cost-apx}
    \end{eqnarray}
    The equality is by definition of the costs $c$ and $\bar{c}$, the first inequality is by definition of the TV-distance and the last inequality is by Lemma~\ref{lem:adv-qptas-1}. 

We now bound the difference in the value functions. Let $p_i\in \{\ell_i,r_i\}$ denote the maximizing choice for $V_i(\vomega)$: so $V_i(\vomega)=c(\vomega) + V_{i+1}(\vomega\circ p_i)$. Then,
\begin{eqnarray*}
    D_i(\pi(\vomega)) & \ge &\bar{c}(\pi(\vomega)) + D_{i+1}(\pi(\vomega\circ p_i)) \,\,\ge \,\,c(\vomega) - \epsilon \, c_{\max} + D_{i+1}(\pi(\vomega\circ p_i)) \\
    &\ge& c(\vomega) - \epsilon \, c_{\max} + V_{i+1}(\vomega\circ p_i)  -(n-i)\cdot \epsilon \, c_{\max}\\
    &=& V_i(\vomega) - (n+1-i)\cdot \epsilon \, c_{\max}.
    \end{eqnarray*}
    Above, the first inequality is by the recurrence for $D_i$, the second inequality is by \eqref{eq:adv-cost-apx}, the third inequality is by induction and the last inequality is by definition of $V_i$. 

    Similarly, if $q_i\in \{\ell_i,r_i\}$ is the maximizing choice for $D_i(\pi(\vomega))$ then: 
    \begin{eqnarray*}
    V_i(\vomega) & \ge & {c}(\vomega) + V_{i+1}(\vomega\circ q_i) \,\,\ge \,\,\bar{c}(\pi(\vomega)) - \epsilon \, c_{\max} + V_{i+1}(\vomega\circ q_i )  \\
    &\ge& \bar{c}(\pi(\vomega))- \epsilon \, c_{\max} + D_{i+1}(\pi(\vomega\circ q_i))  -(n-i)\cdot \epsilon \, c_{\max}\\
    &=& D_i(\pi(\vomega)) - (n+1-i)\cdot \epsilon \, c_{\max}.
    \end{eqnarray*}
    The first inequality is by the recurrence for $V_i$, the second inequality is by \eqref{eq:adv-cost-apx}, the third inequality is by induction and the last inequality is by definition of $D_i$. 
    \end{proof}
Using this lemma with $i=1$, the optimal value of the compressed DP is within an additive error of $n\epsilon\cdot c_{\max} = O(n d \, 2^{-d/2} \, c_{\max})$ of the optimum (given by the exact DP). Therefore, the obtained solution has objective at least 
$$\opt - O(n d \, 2^{-d/2} \, c_{\max})\ge \opt - \frac{c_{\min}}{n}\ge \left(1-\frac1n\right)\cdot \opt,$$
where we set $d = \bigOh{\log\left(n \frac{c_{\max}}{c_{\min}}\right)}$.
This gives us a \gls{qptas}, yielding Theorem~\ref{thm:qptas-adv}.

\bibliographystyle{alpha}
\bibliography{references}

\begin{appendix}
    \section{Bad Example for Greedy with  General Costs}\label{app:extension-to-gen-cost}
We consider a natural extension of the greedy algorithm for unit-cost DR \kofn to the setting with general costs. 
Sort the tests by increasing order of $\frac{c_i}{\ell_i}$, assuming   $k \le  \frac{n}2$.
The proof technique in Section~\ref{sec:unit-cost-kofn} does not work anymore because we cannot ignore the second stopping condition (as we did for $\probb$).
In fact, we can show:

\begin{proposition} \label{prop:cl-ratio-bad-eg}
The greedy algorithm for distributionally-robust \kofn that sorts by increasing $\frac{c_i}{\ell_i}$ has approximation ratio  at least 
$\frac1\epsilon$.
\end{proposition}

\begin{proof}
    Consider an instance with three classes of tests, all tests in this instance have $\ell_i = r_i = p_i$. 
       We set $\gamma=\frac{15}{16}$ and threshold $k=\frac n4$.
    \begin{table}[H]
    \centering

    \begin{tabular}{@{}llll@{}}
    \toprule
    Type & Count & $c_i$ & $p_i$            \\ \midrule
    I    & $0.8 n$ & 0 & $1-\gamma$ \\
    II   & $0.1  n$ & 1 & $\frac12$ \\
    III  & $0.1 n $ & 1 & $1-\epsilon$ \\ \bottomrule
    \end{tabular}
    \end{table}

   Without loss of generality, we can assume that every policy first performs all zero-cost tests (i.e., type I).   
   Note that the greedy algorithm performs test types in the order: I, III, II. We will show that the expected cost of greedy is $\Omega(\frac1\epsilon)$ times that of an ``optimal'' policy which tests in the order I, II, III.

Within  type~I tests, we denote the number of failures  by $F_1$, we know $F_1 \sim \text{Binomial}(0.8n , \gamma)$, with mean $\E[F_1] = \mu:= \frac{3n}4=n-k$ and standard-deviation $\sigma=\sqrt{n\gamma(1-\gamma)}$. Using the Berry-Esseen Theorem, it follows that the cumulative-distribution-function of $F_1$ differs from that of the normal distribution $N(\mu, \sigma)$ by at most $\frac{c}{\sqrt{n}}$, where $c>0$ is some constant. So, we have:
  \begin{equation} \label{eq:binom-eq}
    \Pr\left[ \frac {3n}5   \le F_1 \leq\frac {3n}4  -\sqrt n \right] \ge  \Pr\left[ \frac {3n}5   \le G \leq\frac {3n}4  -\sqrt n \right]  -\frac{2c}{\sqrt{n}} \geq \Omega(1), 
    \end{equation}
  where $G\sim \normaldist(\mu, \sigma^2)$. Under the above event, in order to evaluate the \kofn function we need  (i) at least $\sqrt{n}$ more failures, or (ii) at least $0.05 n$ more successes. 
  
  We can now lower bound the expected cost of greedy:
  $$\mathrm{GRD} \ge \Pr\left[ \frac {3n}5   \le  F_1 \leq\frac {3n}4  -\sqrt n \right] \cdot \E\left[\mbox{ \# type III trials until $\sqrt{n}$ failures}\right] = \Omega(1)\cdot \frac{\sqrt{n}}\epsilon, $$
   where we used the fact that each type III test has failure probability $\epsilon$.  

   We now upper bound the optimal cost by the expected number of type II trials until $(\mu-F_1)^+$ failures. Using the fact that each type II test has failure probability $\frac12$, 
   $$\mathrm{OPT}\le 2\cdot \E[(\mu-F_1)^+]\le 2\cdot \sum_{t\ge 1} \Pr\left[ F_1 \le \mu  - t \right] \le 2\cdot \sum_{t\ge 1}  e^{-2t^2/(4n/5)}=O(\sqrt{n}),$$
   where the last inequality uses Hoeffding's inequality. 
   \end{proof}
We note that the same instance also shows that the greedy algorithm that sorts  by increasing $c_i$ has approximation ratio $\Omega(\frac1\epsilon)$ for the general cost problem; this greedy solution is the same as that in Proposition~\ref{prop:cl-ratio-bad-eg}.

\end{appendix}

\end{document}